\documentclass{article}
\usepackage[english]{babel}
\usepackage[a4paper, total={5.25in, 8in}]{geometry}
\usepackage{graphicx}

\usepackage{tikz-cd}
\usepackage[utf8]{inputenc}
\usepackage[pagewise]{lineno}
\usepackage{fancyhdr}  

\fancypagestyle{firstpage}{
    \fancyhf{}  
    \fancyfoot[L]{2020 MSC Classification: 70H06, 37J30, 58A15, 34A05.}
}
\usepackage{hyperref}
\hypersetup{
    colorlinks=true,
    linkcolor=magenta,
    filecolor=magenta,
    urlcolor=cyan,
    citecolor=magenta,
    }

\usepackage{amssymb, amsmath}
\usepackage{amsthm}
\usepackage{enumitem}

\theoremstyle{remark}

\theoremstyle{definition}
\newtheorem{theorem}{Theorem}[section]

\newtheorem{proposition}{Proposition}[section]

\newtheorem{definition}{Definition}[section]
\newtheorem{remark}{Remark}[section]
\newtheorem{example}{Example}[section]

\newenvironment{keywords}{\setlength{\parindent}{0pt}}{}

\newcommand{\cinf}{$\mathcal{C}^{\infty}$}
\newcommand{\contract}{\,\lrcorner\,}

\usepackage{todonotes}
\usepackage{xcolor}

\usepackage[most]{tcolorbox}

\tcbset{
    mynotestyle/.style={
        breakable,
        enhanced,
        before skip=12pt,
        after skip=12pt,
        left=10pt,
        right=10pt,
        arc=0mm,
        boxrule=1pt,
        parbox=false,
        attach boxed title to top left={xshift=3mm, yshift=-3mm, yshifttext=-1mm},
        boxed title style={sharp corners, size=small},
        fonttitle=\bfseries\sffamily,
    }
}

\newtcolorbox{longpan}{
    mynotestyle, 
    colback=orange!10, 
    colframe=orange!60, 
    title=Pan, 
    colbacktitle=orange!60
}

\newtcolorbox{longzhao}{
    mynotestyle, 
    colback=blue!10, 
    colframe=blue!60, 
    title=Zhao, 
    colbacktitle=blue!60
}

\newtcolorbox{longsardon}{
    mynotestyle, 
    colback=yellow!10, 
    colframe=yellow!60, 
    title=Sardon, 
    colbacktitle=yellow!60,
    coltitle=black 
}

\usepackage{marginnote}
\usepackage{tikz}

\begin{document}

\title{Exact integration of Hamiltonian dynamics via Jacobi and Poisson \texorpdfstring{\cinf}{cinf}-structures}

\author{
Antonio J. Pan-Collantes\\
Department of Mathematics, Universidad de C\'{a}diz, Puerto Real, Spain\\
\texttt{antonio.pan@uca.es}\\
\and
Cristina Sardon\\
Departamento de Matemática Aplicada,\\
Universidad Politécnica de Madrid,\\
Escuela de Edificación, Av. Juan de Herrera 6, 28040 Madrid, Spain\\
\texttt{mariacristina.sardon@upm.es}
\and
Xuefeng Zhao\\
College of Mathematics, Jilin University\\
Changchun 130012, P. R. China\\
\texttt{zhaoxuef@jlu.edu.cn}
}
\maketitle

\thispagestyle{firstpage}  

\begin{abstract}
We develop a geometric framework for the exact integration of Hamiltonian systems based on triangular closure relations among a finite family of functions. Unlike Liouville–Arnold integrability and its noncommutative generalizations, the functions involved in these relations need not be first integrals of the system. Instead, their Hamiltonian vector fields generate a $C^\infty$-structure on phase space that provides an algorithmic procedure for integrating the dynamics.

Within this framework, the equations of motion can be reduced to a finite sequence of completely integrable Pfaffian equations, yielding an explicit integration scheme even when a complete set of conserved quantities is unavailable. The resulting geometric structure is called a \emph{Poisson $C^\infty$-structure}.

We further extend the construction to Jacobi Hamiltonian systems, showing that the same mechanism applies naturally to important subclasses of Jacobi geometry, including Poisson, locally conformally symplectic, and contact manifolds. The method is illustrated on two systems of physical interest: the two-particle non-periodic Toda lattice and the multi-waterbag reduction of the Vlasov equation. We also discuss extensions of the theory to time-dependent Hamiltonian systems.

\end{abstract}

\begin{keywords}
\textbf{Keywords:} Hamiltonian systems, exact solvability, Poisson brackets,
Pfaffian equations, first integrals, \(C^\infty\)-structures, Jacobi manifolds, contact structures, locally conformally symplectic manifolds

\end{keywords}

\section{Introduction}

The classical theory of integrable Hamiltonian systems is rooted in the
Liouville--Arnold theorem \cite{Liouville,Arnold3}. A Hamiltonian system with $n$ degrees of freedom is
said to be completely integrable if it admits $n$ functionally independent first
integrals in involution. Under this assumption, the phase space is foliated by
invariant Lagrangian tori and the dynamics on each torus is quasi-periodic. This
result provides a powerful geometric description of integrability and reduces
the equations of motion, at least in principle, to quadratures. Naturally, other notions of integrability have also been studied, such as the Magri integrability \cite{Fuchssteiner,Magri,Olver} and noncommutative integrability \cite{Azuaje2,Carinena4,Grabowska,Kozlov,Mishchenko1978Generalized,Mishchenko1978Euler}. Further research in the direction of integrability has been extensively developed in numerous studies \cite{Carlet,Carlet2,Dubrovin,Lorenzoni2,Zhao2,Zhao3}.

A major extension of this framework is provided by the theory of noncommutative
integrability introduced by Mishchenko and Fomenko \cite{Mishchenko1978Generalized,Mishchenko1978Euler}. In this setting, one allows
a sufficiently large family of first integrals that need not commute, provided
their Poisson algebra satisfies an appropriate completeness condition. Although commutativity is
relaxed, the dynamics remains strongly constrained and is again organized by
invariant isotropic tori.

Despite their fundamental role, both Liouville--Arnold integrability and its
noncommutative generalizations are essentially \emph{structural} notions. They
guarantee the existence of conserved quantities restricting the dynamics, but
they do not automatically yield explicit solutions. In particular, integrability
does not imply \emph{exact solvability}, understood as the possibility of
constructing trajectories by a finite sequence of explicit local integrations. Even
for Liouville-integrable systems, the construction of action--angle variables \cite{Duistermaat1980,Nehorosev1972}
or the evaluation of the required quadratures may be analytically intractable.

The present work addresses this gap by focusing directly on exact solvability.
Our guiding principle is that explicit integration should not be tied
exclusively to the existence of first integrals. Instead, we show that the
dynamics of a Hamiltonian system can be organized and explicitly integrated by
means of an ordered set of functions, even when these functions are not
constants of motion.

A closely related and complementary line of research has been developed by Kresic-Juric, Muriel and Ruiz \cite{kresicjuric2025}, who construct canonical solvable structures for completely integrable Hamiltonian systems. Starting from $n$ first integrals in involution, they build $2n$ Hamiltonian vector fields that form a solvable structure and show that the associated Pfaffian forms recover the classical action--angle variables, providing a novel and illuminating geometric interpretation of the Arnold--Liouville theorem.

The present work is motivated by a different but related question: whether explicit integration can be achieved even when a complete set of conserved quantities is not available. The triangular closure condition defining a Poisson $C^\infty$-structure allows functions that evolve nontrivially along the Hamiltonian flow to be systematically incorporated into the integration scheme, and the resulting procedure yields integration via completely integrable Pfaffian equations. Together, the two approaches offer complementary perspectives on Hamiltonian integrability: solvable structures exploit the full power of conservation laws to achieve integration by quadratures, while Poisson $C^\infty$-structures trade the quadrature property for applicability to a broader class of systems, including those that are not Liouville integrable and dynamics on Jacobi manifolds such as contact and locally conformally symplectic geometries.

More precisely, we introduce a structure based on the existence of an ordered
family of $2n-2$ functionally independent functions
\[
\mathcal{F} = ( f_1, \dots, f_{2n-2} ),
\]
defined on a $2n$-dimensional symplectic manifold together with a Hamiltonian
function $H\in\mathcal C^\infty(M)$. For notational convenience, we will set
\[
 f_0:=H.
\]
The defining feature of the family is a triangular closure condition with
respect to the Poisson bracket: whenever $j>i$ (allowing $i=0$), the bracket
$\{f_j,f_i\}$ depends only on the functions $f_0,\dots,f_j$. This condition
induces a natural hierarchy in the Poisson algebra of functions and generalizes
both commutative and noncommutative integrability without requiring the
functions involved to be conserved.

Unless explicitly stated otherwise, functional independence is understood on an
open dense subset $U\subseteq M$ (equivalently, $df_1\wedge\cdots\wedge df_{2n-2}\neq 0$
on $U$). All constructions and integration statements are meant on $U$ and are,
in particular, local near any regular point.

We refer to such a family as a \emph{Poisson \cinf-structure}. The central
result of this paper shows that the existence of a Poisson \cinf-structure
induces a \cinf-structure, in the sense of \cite{pancinf-sym}, for the
distribution generated by the relevant Hamiltonian vector field. As a consequence, the
equations of motion can be integrated, at least locally, by means of a sequence of
completely integrable Pfaffian equations. In this way, the Poisson
\cinf-structure provides not only an integrability criterion, but also a
constructive integration algorithm.

The framework developed here applies naturally to Hamiltonian systems on symplectic manifolds, where the Poisson structure is nondegenerate. However, many physically and geometrically relevant Hamiltonian systems arise in more general settings where the underlying structure is not symplectic. In particular, Jacobi manifolds provide a natural generalization of Poisson geometry and include important subclasses such as Poisson, locally conformally symplectic, and contact manifolds.

Motivated by this observation, we extend the notion of Poisson $C^\infty$-structures to the broader context of Jacobi Hamiltonian systems. In this setting, the presence of the Reeb vector field introduces additional geometric features, but the triangular closure mechanism still yields a $C^\infty$-structure for the Hamiltonian distribution and leads to the same Pfaffian integration procedure.

The paper is organized as follows. Section~\ref{sec:preliminaries} reviews the necessary background on symplectic geometry, Hamiltonian systems, and \cinf-structures for distributions. Section~\ref{poissoncinf} introduces Poisson \cinf-structures, establishes the main integration theorem, and develops the integration algorithm; it also treats time-dependent Hamiltonian systems and illustrates the construction on the two-particle non-periodic Toda lattice and the multi-waterbag reduction of the Vlasov equation. Section~\ref{sec:jacobicinf} extends the framework to Jacobi Hamiltonian systems, covering the Poisson, locally conformally symplectic, and contact subcases. Section~5 collects conclusions and directions for future work.

\section{Preliminaries}\label{sec:preliminaries}
In this section, we briefly review the geometric framework of Hamiltonian mechanics and the Pfaffian integration method based on \cinf-structures. We also introduce the necessary notation and conventions used throughout the paper.

\subsection{Symplectic geometry and Hamiltonian systems}
Let $(M, \omega)$ be a symplectic manifold of dimension $2n$ \cite{Arnold3,Arnold2,Zhao1}. The non-degeneracy of the closed 2-form $\omega$ induces a vector bundle isomorphism between $TM$ and $T^*M$ via the contraction map $v \mapsto v \contract \omega$. Consequently, for any function $f \in \mathcal{C}^{\infty}(M)$, the Hamiltonian vector field $X_f$ is uniquely defined by $X_f \contract \omega = df$, and given functions $f_1, \dots, f_k$, their associated Hamiltonian vector fields $X_{f_1}, \dots, X_{f_k}$ are pointwise linearly independent on the open set where $f_1,\dots,f_k$ are functionally independent.

The symplectic form naturally endows $\mathcal{C}^{\infty}(M)$ with a Lie algebra structure via the Poisson bracket $\{f, g\} = \omega(X_f, X_g) = X_g(f)$. The mapping $f \mapsto X_f$ constitutes a Lie algebra homomorphism, satisfying $[X_f, X_g] = X_{\{f, g\}}$, which links the algebraic structure of the observables to the geometry of the flow.

Recall that a symplectic manifold comes equipped with a natural volume form, the Liouville volume form, given by:
\begin{equation}\label{liouvilleform}
    \Omega = \frac{\omega^n}{n!},
\end{equation}
where $\omega^n$ denotes the $n$-th exterior power of $\omega$.

In the context of symplectic geometry, the Pfaffian \cite{Bourbaki,Cayley,Vein} provides a convenient way to express the contraction of vectors with the symplectic volume form. Specifically, for any set of $2n$ vector fields $Y_1, \dots, Y_{2n}$, we have:
\begin{equation}\label{eq:pfaffianvol}
\Omega(Y_1,\dots,Y_{2n}) = \operatorname{Pf}(\mathbf{W}), 
\end{equation}
where $\mathbf{W}$ is the $2n \times 2n$ skew-symmetric matrix with entries $w_{ab} = \omega(Y_a, Y_b)$. The \emph{Pfaffian} appearing in this formula is a polynomial associated with any even-dimensional skew-symmetric matrix $A = (a_{ij})$. It satisfies $(\operatorname{Pf}(A))^2 = \det(A)$ and is defined by
\[
\operatorname{Pf}(A) := \frac{1}{2^m m!} \sum_{\sigma \in S_{2m}} \mathrm{sgn}(\sigma) \prod_{k=1}^m a_{\sigma(2k-1),\,\sigma(2k)}.
\]
For odd-dimensional matrices, the Pfaffian is defined to be zero. It is a well-known fact that the Pfaffian also admits the following expansion for any fixed index $j$:
\begin{equation}\label{pfaffianexpansion}
    \operatorname{Pf}(A) = \sum_{\substack{i=1 \\ i \neq j}}^{2m} (-1)^{i+j+1} \, a_{ij} \, \operatorname{Pf}\big(A_{\widehat{i}\,\widehat{j}}\big),
\end{equation}
where $A_{\widehat{i}\,\widehat{j}}$ denotes the skew-symmetric matrix obtained from $A$ by deleting the $i$-th and $j$-th rows and columns \cite{Gantmacher}.

\subsection{\texorpdfstring{\cinf}{cinf}-structures of distributions}\label{prelCinf-structures}

We recall the notions of \cinf-symmetries and \cinf-structures for distributions developed in \cite{pancinf-sym,pancinf-struct,pan23integration}. These concepts generalize classical symmetries and solvable structures.

A \emph{distribution} $\mathcal{D}$ on a manifold $M$ (of dimension $m$) is a $\mathcal{C}^{\infty}(M)$-submodule of the module of vector fields $\mathfrak{X}(M)$. We assume the distribution has constant rank, meaning the dimension of the subspaces $\mathcal{D}_p := \{X(p) \mid X \in \mathcal{D}\} \subset T_pM$ is constant.
A submanifold $N\subset M$ is called an \emph{integral manifold} of $\mathcal{D}$
if $T_pN=\mathcal{D}_p$ for all $p\in N$.
\begin{definition}
Let $\mathcal{D}$ be a distribution on a manifold $M$. A vector field $X \in \mathfrak{X}(M)$ is called a \cinf-symmetry of $\mathcal{D}$ if $X\notin \mathcal{D}$ and
\begin{equation}
    [X, \mathcal{D}] \subset \mathcal{D} \oplus \langle X \rangle,
\end{equation}
where $\langle ... \rangle$ denotes the $\mathcal{C}^{\infty}(M)$-module generated by the vector fields inside the brackets, and $\oplus$ denotes the direct sum of modules.

In other words, for any vector field $Y \in \mathcal{D}$, the Lie bracket $[X, Y]$ can be expressed as a linear combination of vector fields in $\mathcal{D}$ and $X$ itself, with coefficients in $\mathcal{C}^{\infty}(M)$.
\end{definition}

\begin{definition}\label{def1}
Let $\mathcal{D}$ be a distribution of rank $k$ on an $m$-dimensional manifold $M$. An ordered set of vector fields $(X_1, X_2, \dots, X_{m-k})$ is called a \cinf-structure for $\mathcal{D}$ if the following conditions hold:
\begin{enumerate}
    \item $X_1$ is a \cinf-symmetry of $\mathcal{D}$.
    \item For each $j > 1$, $X_j$ is a \cinf-symmetry of the distribution $\mathcal{D}_{j-1}$ defined by
    \begin{equation}
        \mathcal{D}_{j-1} = \mathcal{D} \oplus \langle X_1, \dots, X_{j-1} \rangle.
    \end{equation}
\end{enumerate}
\end{definition}

\begin{remark}\label{rem:last}
Observe that the final vector field of a \cinf-structure may be chosen to be any vector field that is pointwise linearly independent of the preceding fields in the structure and the generators of the distribution to be integrated.
\end{remark}

\begin{remark}
    Given a distribution $\mathcal{D} = \langle Y_1, \dots, Y_k \rangle$, a \cinf-structure for $\mathcal{D}$ can be alternatively characterized as an ordered set of vector fields $(X_1, X_2, \dots, X_{m-k})$ such that the vector fields $X_1, \dots, X_{m-k}, Y_1, \dots, Y_k$ are pointwise linearly independent in an open subset and the distributions $\mathcal{D}_j$ are involutive, $1 \leq j \leq m-k$. 
\end{remark}

The fundamental result connecting these structures to integrability is the following theorem:

\begin{theorem}\cite[Theorem 3.5]{pancinf-sym}\label{TZ}
Let $\mathcal{D}$ be a distribution of rank $k$ on an $m$-dimensional manifold $M$. If $\mathcal{D}$ admits a \cinf-structure $(X_1, \dots, X_{m-k})$, then the integral manifolds of $\mathcal{D}$ can be obtained by solving a sequence of $m-k$ completely integrable Pfaffian equations. 
\end{theorem}

The proof of this theorem provides a constructive procedure to perform the local integration of the distribution
$\mathcal{D}$. Explicitly, if $\Omega$ is any volume form on $M$ and $(X_1, \ldots, X_{m-k})$ is a \cinf-structure for $\mathcal{D}$, we obtain $m-k$ linearly independent 1-forms $\omega_i$ ($i=1,\ldots,m-k$) by
\begin{equation}\label{formscontract}
    \omega_i = X_{m-k}\contract \cdots \widehat{X_i} \cdots \contract X_{1} \contract Y_k \contract\cdots \contract Y_1 \contract \Omega,
\end{equation}
where $Y_1,\ldots,Y_k$ are vector fields generating $\mathcal{D}$, and the hat denotes omission of $X_i$. These 1-forms satisfy:
\begin{equation}\label{cinfcondforms}
    d\omega_i=0 \text{ mod } \langle\omega_i,\omega_{i+1},\ldots,\omega_{m-k}\rangle.
\end{equation}

Once explicit formulas for the Pfaffian 1-forms $\omega_i$ are available, the integration procedure consists of successively solving the associated Pfaffian equations, starting with the last 1-form in the sequence:
$$
\omega_{m-k}\equiv 0.
$$
The solution lets us reduce to a lower-dimensional level set, and the process is continued recursively: on each level set of previously integrated equations, the next equation is solved. This finally yields a parametrization of the integral manifolds of the original distribution.  The interested reader can find the details in \cite{pancinf-sym,pancinf-struct,pan23integration}.

\section{Poisson \texorpdfstring{\cinf}{cinf}-structures and integration}\label{poissoncinf}

In this section we introduce the main new structure of the paper and show how it
leads to a constructive integration procedure for Hamiltonian systems. Unlike
classical integrability theories, the approach developed here does not rely on
the knowledge of families of first integrals, but instead on an ordered set of functions with a closure property with respect to the Poisson bracket.

Consider a Hamiltonian system $(M,H)$ consisting of a symplectic manifold $(M,\omega)$ of dimension $2n$, and a distinguished function $H\in \mathcal{C}^\infty(M)$. The associated Hamiltonian vector field $X_H$ generates the distribution
\[
\mathcal{D}_H = \langle X_H \rangle.
\]
Integrating the Hamiltonian dynamical system amounts to determining
the integral manifolds (curves) of $\mathcal{D}_H$.

Within this setting, we introduce the following concept.

\begin{definition}\label{def:poisson-cinf}  
A \emph{Poisson \cinf-structure} for the Hamiltonian system $(M,H)$ is an
ordered family of $2n-2$ functions
\[
\mathcal{F} = (f_1,\dots,f_{2n-2}), \quad f_i \in\mathcal{C}^\infty(M), \quad 1 \leq i \leq 2n-2,
\]
such that, taking $f_0:=H$, the set $\{f_0,f_1,\dots,f_{2n-2}\}$ is functionally independent, 
and for $0\le i<j\le 2n-2$, the Poisson bracket satisfies
\begin{equation}\label{condicion}
    \{f_j,f_i\} = F_{ji}(f_0,f_1,\dots,f_j),
\end{equation}
for some smooth functions $F_{ji}$.
\end{definition}

The knowledge of a Poisson \cinf-structure plays a significant role in the integration of the Hamiltonian system, as stated in the following theorem.

\begin{theorem}\label{thm:main}
Consider a Hamiltonian system $(M,H)$ on a $2n$-dimensional symplectic manifold $(M,\omega)$. Let $\mathcal{F}=(f_1,\dots,f_{2n-2})$ be a Poisson \cinf-structure for $(M,H)$. Then the Hamiltonian distribution
\[
\mathcal{D}_H=\langle X_H\rangle
\]
admits the \cinf-structure 
\[
(X_{f_1},X_{f_2},\dots,X_{f_{2n-2}},R),
\]
where $R$ is an arbitrary vector field, such that $R\notin \langle X_H, X_{f_1},\dots,X_{f_{2n-2}}\rangle$.
\end{theorem}

\begin{proof}
\sloppy
Let $\mathcal{F}=(f_1,\dots,f_{2n-2})$ be a Poisson \cinf-structure, and set $f_0:=H$.
Since the functions $f_0,f_1,\dots,f_{2n-2}$ are functionally independent, the associated Hamiltonian
vector fields 
$$
X_{f_0}, X_{f_1},X_{f_2},\dots,X_{f_{2n-2}}
$$
are pointwise linearly independent.

We must show that the Lie brackets of these fields satisfy the condition required to be a \cinf-structure for $\mathcal{D}_H$. By the definition of a Poisson \cinf-structure, for any $j > i$, the Poisson bracket depends only on the preceding functions:
\[
\{f_j, f_i\} = F_{ji}(f_0, f_1, \dots, f_j)
\]
for some smooth function $F_{ji}$. Applying the chain rule, the differential of this bracket is
\[
d\{f_j, f_i\} = \sum_{k=0}^j \frac{\partial F_{ji}}{\partial f_k} df_k.
\]
Using the isomorphism between 1-forms and Hamiltonian vector fields, we obtain
\[
X_{\{f_j,f_i\}} = \sum_{k=0}^j \frac{\partial F_{ji}}{\partial f_k} X_{f_k}.
\]
This implies that
\[
[X_{f_j},X_{f_i}] \in \langle X_{f_0}, X_{f_1}, \dots, X_{f_j} \rangle.
\]
Thus, each $X_{f_j}$ is a \cinf-symmetry of the distribution generated by $X_H$ together with $(X_{f_1},\dots,X_{f_{j-1}})$. The ordered family $(X_{f_1},X_{f_2},\dots,X_{f_{2n-2}})$, completed with a pointwise linearly independent vector field $R$ (see Remark~\ref{rem:last}), therefore defines a \cinf-structure for $\mathcal{D}_H$ in the sense of Section~\ref{sec:preliminaries}.

The result now follows directly from Theorem~\ref{TZ}, which
guarantees that the integral manifolds of $\mathcal{D}_H$ can be obtained by
solving a sequence of completely integrable Pfaffian equations.
\end{proof}

\begin{remark}
The functions defining a Poisson \cinf-structure, while enabling integration, are not required to be
constants of motion. In particular, for $j\geq 1$, the functions $f_j$ may evolve
nontrivially along the Hamiltonian flow. This feature distinguishes the present
framework from Liouville--Arnold and Mishchenko--Fomenko integrability, which are
based exclusively on families of conserved quantities.
\end{remark}

As a consequence of Theorem~\ref{thm:main}, once a Poisson \cinf-structure is known for a Hamiltonian system, the integration of the equations of motion can be performed by means of a sequence of $2n-1$ completely integrable Pfaffian equations. The integration procedure, as described in Section \ref{prelCinf-structures}, requires the explicit construction of 1-forms given by equation \eqref{formscontract}. In the rest of this section, we will derive a more explicit expression for these 1-forms, directly in terms of the functions defining the Poisson \cinf-structure. 

For this purpose, we introduce an auxiliary and arbitrary smooth function, functionally independent of the family $(H, f_1, \dots, f_{2n-2})$, so its corresponding Hamiltonian field plays the role of $R$ in the proof of Theorem \ref{thm:main}. We will denote $f_0:=H$, and $f_{2n-1}$ the auxiliary function, for notational convenience.
Then, the procedure outlined in Section \ref{prelCinf-structures} involves the 1-forms:
\begin{equation}\label{originalforms}
    \omega_i = X_{f_{2n-1}} \contract X_{f_{2n-2}} \contract \cdots \, \widehat{X_{f_i}} \, \cdots \contract X_{f_{1}} \contract X_{f_0} \contract \Omega, \quad 1\leq i \leq 2n-1.
\end{equation}

Let $Z$ be an arbitrary vector field on $M$. For $1\leq i \leq 2n-1$, we have
$$
\omega_i(Z) = \Omega(X_{f_0}, X_{f_1}, \ldots,\widehat{X_{f_i}},\ldots, X_{f_{2n-1}}, Z),
$$
where we consider $\Omega$ the Liouville volume form defined in \eqref{liouvilleform}. By using equation \eqref{eq:pfaffianvol}, we have
$$
\omega_i(Z)=\operatorname{Pf}(\mathbf{M}_i),
$$
where $\mathbf{M}_i$ is the $2n \times 2n$ skew-symmetric matrix given by the symplectic products of the vector fields 
$$
X_{f_0}, X_{f_1}, \ldots, \widehat{X_{f_i}}, \ldots, X_{f_{2n-1}}, Z,
$$
arranged in that order. To analyze the structure of $\mathbf{M}_i$ in detail, we consider the $2n\times 2n$ matrix $\mathbf{F} = (w_{ab})$ with entries defined by the Poisson brackets:
$$
w_{ab} = \{f_{a-1}, f_{b-1}\}=F_{a-1,b-1},
$$
for certain smooth functions $F_{a-1,b-1}$.

For the arbitrary vector field $Z$, we have $\omega(X_{f_k}, Z) = df_k(Z)$. Therefore, the matrix $\mathbf{M}_i$ has the block structure
$$
\mathbf{M}_i = 
\left(
\begin{array}{c|c}
\mathbf{F}_{\widehat{i}} & d\mathbf{f}(Z)_{\widehat{i}} \\
\hline
-d\mathbf{f}(Z)_{\widehat{i}}^T & 0
\end{array}
\right),
$$
where $\mathbf{F}_{\widehat{i}}$ denotes the matrix $\mathbf{F}$ with the $(i+1)$-th row and column removed, and $d\mathbf{f}(Z)_{\widehat{i}}$ denotes the column vector of values $df_k(Z)$ for $k \in \{0,\dots,2n-1\}\setminus\{i\}$.

Now, we use equation \eqref{pfaffianexpansion} to expand $\operatorname{Pf}(\mathbf{M}_i)$ along its last column. The expansion yields
\begin{equation}
    \omega_i(Z) = \sum_{\substack{k=0 \\ k \neq i}}^{2n-1} (-1)^{k} \operatorname{sgn}(i-k) \operatorname{Pf}(\mathbf{F}_{\widehat{i,k}}) \, df_k(Z),
\end{equation}
where $\mathbf{F}_{\widehat{i,k}}$ denotes the matrix $\mathbf{F}$ with rows and columns $i+1$ and $k+1$ removed.

We can summarize the findings, for further reference, in the following theorem:
\begin{theorem}\label{thm:pfaffinaformula}
Let $(M,H)$ be a Hamiltonian system admitting a Poisson \cinf-structure $\mathcal{F} = (f_1, \ldots, f_{2n-2})$, and set $f_0:=H$. Let $f_{2n-1}\in \mathcal{C}^\infty(M)$ be functionally independent of $f_0,f_1,\dots,f_{2n-2}$. Then the system can be locally integrated by solving the Pfaffian equations generated by the 1-forms:
\begin{equation}\label{eq:omega_i_unified}
    \omega_i = \sum_{\substack{k=0 \\ k \neq i}}^{2n-1} (-1)^{k} \operatorname{sgn}(i-k) \operatorname{Pf}(\mathbf{F}_{\widehat{i,k}}) \, df_k, \quad 1 \leq i \leq 2n-1.
\end{equation}
\end{theorem}

\begin{example}[Case $n=2$]\label{casen2}
    Consider a Hamiltonian system of dimension $2n=4$. In this case, a Poisson \cinf-structure consists of two functions $\mathcal{F}=(f_1,f_2)$. Setting $f_0:=H$ and introducing an auxiliary function $f_3$, the general formula \eqref{eq:omega_i_unified} explicitly yields:
\begin{align*}
\omega_1 &= F_{23} \, df_0 - F_{03} \, df_2 + F_{02} \, df_3, \\
\omega_2 &= F_{13} \, df_0 - F_{03} \, df_1 + F_{01} \, df_3, \\
\omega_3 &= F_{12} \, df_0 - F_{02} \, df_1 + F_{01} \, df_2.
\end{align*}
\end{example}


\subsection{Integration algorithm}\label{sec:integration-algorithm}
According to the results of the previous section, we propose the following constructive algorithm for the explicit integration of a Hamiltonian system $(M, H)$ of dimension $2n$.
\begin{enumerate}[label=\textbf{Step \arabic*.}, leftmargin=*]
        
    \item Identify a Poisson $\mathcal{C}^{\infty}$-structure $\mathcal{F}=(f_{1}, \dots, f_{2n-2})$ for $(M,H)$ and set $f_0:=H$.
        
    \item Choose an arbitrary smooth function $f_{2n-1} \in \mathcal{C}^{\infty}(M)$ functionally independent of the set $\{f_0,f_1,\dots,f_{2n-2}\}$.
        
    \item Define the $2n \times 2n$ skew-symmetric matrix $\mathbf{F}$ with entries defined by the Poisson brackets:
        \begin{equation}
             w_{k+1,j+1}= \{f_k, f_j\}=F_{kj}, \quad 0 \le k, j \le 2n-1.
        \end{equation}
        
    \item Construct the sequence of $2n-1$ Pfaffian 1-forms $(\omega_1, \dots, \omega_{2n-1})$ using the matrix $\mathbf{F}$ and the formula:
        \begin{equation}\label{eq:omega_i_unified_repeated}
            \omega_i = \sum_{\substack{k=0 \\ k \neq i}}^{2n-1} (-1)^{k} \operatorname{sgn}(i-k) \operatorname{Pf}(\mathbf{F}_{\widehat{i,k}}) \, df_k, 
        \end{equation}
        for $1 \leq i \leq 2n-1$.

    \item Solve the Pfaffian equation $\omega_{2n-1} \equiv 0$, i.e., find a particular solution to the linear homogeneous first-order partial differential equations system in the unknown $I_{2n-1}$ given by:
        $$
        \omega_{2n-1} \wedge dI_{2n-1} = 0.
        $$

    \item Iterate: for $k = 2n-2, \dots, 1$, restrict the remaining 1-forms to any regular level set of the previously found integrals $\{I_{k+1}, \dots, I_{2n-1}\}$ and solve $\omega_{k} \equiv 0$ to find the next function $I_{k}$.
    
    \item The sequence of functions $I_1, \dots, I_{2n-1}$ provides the implicit description of the distribution $\mathcal{D}_H$.
\end{enumerate}

\begin{remark}
The algorithm does not require the explicit computation of action--angle
variables, invariant tori, or additional conserved quantities beyond the
Hamiltonian itself. Its effectiveness relies on the finding of a
Poisson \cinf-structure and the subsequent solution of the sequence of Pfaffian equations.
\end{remark}

\begin{example}
\sloppy
We illustrate the integration algorithm with the two-particle non-periodic Toda lattice. Consider the phase space $M = T^*\mathbb{R}^2 \cong \mathbb{R}^4$ with canonical coordinates 
$(q_1,q_2,p_1,p_2)$ and Hamiltonian
\[
H=\frac12 p_1^2+\frac12 p_2^2+e^{q_1-q_2}.
\]
Introducing center-of-mass and relative coordinates
\[
Q=\frac{q_1+q_2}{2}, \qquad q=q_1-q_2, \qquad P=p_1+p_2, \qquad p=\frac{p_1-p_2}{2},
\]
the transformation is symplectic and the Hamiltonian takes the form
\[
H=p^2+\frac{P^2}{4}+e^{q}.
\]
We define the ordered family
\[
\mathcal{F}=(f_1,f_2) \qquad\text{by}\qquad f_1=P,\quad f_2=Q,
\]
and set $f_0:=H$. These functions are functionally independent. Their Poisson brackets satisfy
\begin{align*}
\{f_1,f_0\} &= \{P, H\} = 0, \\
\{f_2,f_0\} &= \{Q, H\} = \frac{\partial H}{\partial P} = \frac{P}{2} = \frac{1}{2}f_1, \\
\{f_2,f_1\} &= \{Q, P\} = 1.
\end{align*}
Observe that even if $f_1$ is a conserved quantity, $f_2$ is not. But condition \eqref{condicion} is satisfied, so $\mathcal{F}$ constitutes a Poisson \cinf-structure.

To apply the integration algorithm, we choose the auxiliary function $f_3=p$, which is functionally independent of $\{f_0,f_1,f_2\}$. The skew-symmetric matrix $\mathbf{F}=(\{f_i,f_j\})$ is given by
\[
\mathbf{F} = 
\begin{pmatrix}
0 & 0 & -P/2 & e^q \\
0 & 0 & -1 & 0 \\
P/2 & 1 & 0 & 0 \\
-e^q & 0 & 0 & 0
\end{pmatrix}.
\]

\sloppy
Following the algorithm, we now compute explicitly the three Pfaffian 1-forms $\omega_1, \omega_2, \omega_3$, using formula \eqref{eq:omega_i_unified_repeated}:
\begin{align*}
\omega_1 &= -\frac{P}{2} dp - e^q dQ, \\
\omega_2 &= -e^q dP, \\
\omega_3 &= -dH +\frac{P}{2} dP.
\end{align*}

We now proceed to solve the associated Pfaffian equations sequentially. The equation $\omega_3 = 0$ yields the first integral:
\[
I_3 = H - \frac{1}{4}P^2 = p^2 + e^q,
\]
which represents the energy of the relative motion. We restrict the system to the level set $\Sigma_{c_3} = \{x \in M \mid I_3(x) = c_3\}$. On this level set, we solve for $q$ (specifically $e^q$) in terms of $p$, which implies the substitution:
\begin{equation}\label{relation}
    e^q = c_3 - p^2.
\end{equation}
Restricting $\omega_2$ to this level set by substituting $e^q$, we obtain:
\[
\omega_2|_{\Sigma_{c_3}} = -(c_3 - p^2) dP.
\]
In any open neighborhood where $c_3 - p^2 \neq 0$, the equation $\omega_2|_{\Sigma_{c_3}} \equiv 0$ implies $dP \equiv 0$. This yields the second integral $I_2 = P$. We further restrict the 1-form $\omega_1$ to the submanifold $\Sigma_{c_3, c_2} = \Sigma_{c_3} \cap \{P = c_2\}$ by substituting $P=c_2$ and $e^q = c_3 - p^2$:
\[
\omega_1|_{\Sigma_{c_3, c_2}} = -\frac{c_2}{2} dp - (c_3 - p^2) dQ.
\]
Setting this form to zero yields a separable differential equation:
\[
-\frac{c_2}{2} dp = (c_3 - p^2) dQ.
\]

Integrating this equation yields the solution for $Q$, and consequently the integral curves of the Hamiltonian distribution:
\begin{align*}
q(p)&= \ln (c_3 - p^2) \\
Q(p) &= -\frac{c_2}{2\sqrt{c_3}} \operatorname{arctanh}\left(\frac{p}{\sqrt{c_3}}\right) + c_1\\
P(p) &= c_2 \\
\end{align*}

Now, we can recover the time dependency by substituting the derived relation \eqref{relation} into the equation of motion for $p$:
\begin{equation}
    \dot{p} = - \frac{\partial H}{\partial q} = -e^q = -(c_3 - p^2).
\end{equation}

The full solution set $(q(t), Q(t), P(t))$ follows immediately by substituting $p(t)$ back into the geometric relations derived previously.
   
\end{example}

\begin{remark}
The two–particle non–periodic Toda lattice is well known to be Liouville integrable and admits explicit solutions by classical methods, including separation of variables and action--angle coordinates. The purpose of the present example is not to rederive these results, but to show that the dynamics can be integrated \emph{directly} from a Poisson \cinf-structure, without appealing to invariant tori, commuting families of first integrals, or canonical transformations to normal forms. In this sense, the example serves as a benchmark illustrating how the proposed framework recovers exact solvability by a fundamentally different mechanism.
\end{remark}

\subsection{Time-dependent Hamiltonian systems}\label{sec:time-dependent}

A time-dependent Hamiltonian system with $n$ degrees of freedom is described by coordinates $(q_i,p_i)$, which we denote collectively by $(q,p)$. The dynamics is defined by a Hamiltonian function $H(q,p,t)$ that explicitly depends on time $t$. Such systems can be formulated as an autonomous Hamiltonian system on the extended phase space
with coordinates $(t,q,E,p)$, where $E$ plays the role of the variable conjugate to
time $t$. The extended phase space has dimension $2(n+1)$ and is equipped with the canonical symplectic form
\[
\omega_{\mathrm{ext}} = \omega+dt \wedge dE,
\]
with a Poisson bracket defined by
\[
\{t,E\}=1,\qquad \{q_i,p_j\}=\delta_{ij},
\]
all others vanishing. 

The extended Hamiltonian is given by
\[
H_{\mathrm{ext}}(q,p,t,E)=H(q,p,t)+E,
\]
and the Hamiltonian vector field of $H_{\mathrm{ext}}$
generates both the original dynamics and the trivial evolution
$\dot t=1$.

Since the extended phase space has dimension $2(n+1)$, a Poisson
\cinf-structure would, in
principle, require an ordered family of $2(n+1)-2=2n$ functionally
independent functions (in addition to the Hamiltonian $H_{\mathrm{ext}}$).

A key observation is that the time variable $t$ always satisfies the Poisson relations
\[
\{t,H_{\mathrm{ext}}\}=\{t,E\}=1,
\]
\[
\{g,t\}=0,
\]
for any function $g=g(q,p,t)$ independent of $E$.
As a consequence, $t$ can always be chosen as the first nontrivial
element of a Poisson \cinf-structure for the extended system, provided that we restrict to functions independent of $E$. This
reduces the construction problem to the search for $2n-1$
additional functions.

Thus, the transition to time-dependent Hamiltonians requires no modification to the underlying algebraic framework. Moreover, the increase in the phase space dimension by two does not correspond to the search for two additional functions in the Poisson \cinf-structure, since one of them is always provided by the time variable itself.

\begin{example}
We now present an explicit example of a time-dependent Hamiltonian
system that is exactly solvable via a Poisson \cinf-structure. Consider a one degree of freedom time-dependent Hamiltonian system with Hamiltonian
\[
H(q,p,t)=\frac{1}{2}p^2+q\,t.
\]
Passing to the extended phase space with coordinates $(t,q,E,p)$, the
associated extended Hamiltonian is
\[
H_{\mathrm{ext}}=\frac{1}{2}p^2+q\,t+E,
\]
with canonical Poisson brackets
\[
\{q,p\}=1,\qquad \{t,E\}=1,
\]
all others vanishing. The Hamiltonian vector field of $H_{\mathrm{ext}}$ is
\[
X_{H_{\mathrm{ext}}}
=
p\,\frac{\partial}{\partial q}
-
t\,\frac{\partial}{\partial p}
+
\frac{\partial}{\partial t}
-
q\,\frac{\partial}{\partial E},
\]
and the equations of motion are
\[
\dot q=p,\qquad
\dot p=-t,\qquad
\dot t=1,\qquad
\dot E=-q.
\]

Since the extended phase space has dimension $4$, a Poisson
\cinf-structure consists of two functionally independent functions.
We define the ordered family
\[
\mathcal F=(f_1,f_2)
\qquad\text{with}\qquad
f_1=t,\quad
f_2=p,
\]
and set $f_0:=H_{\mathrm{ext}}$. These functions are functionally independent and satisfy
\[
\{f_1,f_0\}=1,\qquad
\{f_2,f_0\}=-t=-f_1,\qquad
\{f_2,f_1\}=0.
\]
Hence, the triangular closure condition of
Definition~\ref{def:poisson-cinf} is satisfied and $\mathcal F$ defines a Poisson \cinf-structure for the extended Hamiltonian system. Consider now the auxiliary function $f_3=q$, which is functionally independent of $\{f_0,f_1,f_2\}$. 

Following the algorithm, we compute the Pfaffian 1-forms using formula \eqref{eq:omega_i_unified_repeated}. The explicit calculation yields:
\begin{align*}
    \omega_1 &= -q \, dt - dE \\
    \omega_2 &= p \, dt - dq \\
    \omega_3 &= -t \, dt - dp
\end{align*}

We now proceed to solve the associated Pfaffian equations sequentially, starting from $\omega_3 \equiv 0$. We have $dp = -t \, dt$, whose integration yields the first integral
\begin{equation}
    I_3 = p + \frac{1}{2}t^2 = c_3,
\end{equation}
so that the momentum evolves as $p(t) = c_3 - \frac{1}{2}t^2$. 

Next, we substitute this expression for $p(t)$ into the equation $\omega_2 \equiv 0$, obtaining $dq = (c_3 - \frac{1}{2}t^2) \, dt$. Integration gives the restriction
\begin{equation}
    I_2 = q - c_3 t + \frac{1}{6}t^3 = c_2,
\end{equation}
which implies $q(t) = c_2 + c_3 t - \frac{1}{6}t^3$. 

Finally, we substitute $q(t)$ into $\omega_1 \equiv 0$, which becomes $dE = -(c_2 + c_3 t - \frac{1}{6}t^3) \, dt$. Integrating for the energy variable $E$ yields
\begin{equation}
    I_1 = E + c_2 t + \frac{1}{2}c_3 t^2 - \frac{1}{24}t^4 = c_1.
\end{equation}

Thus, the Poisson $C^\infty$-structure allows us to recover the exact solution of the time-dependent system. The integral curves of $X_{H_{\mathrm{ext}}}$ are explicitly given by:
\begin{equation}
\begin{aligned}
    p(t) &= -\frac{1}{2}t^2 + c_3, \\
    q(t) &= -\frac{1}{6}t^3 + c_3 t + c_2, \\
    E(t) &= \frac{1}{24}t^4 - \frac{1}{2}c_3 t^2 - c_2 t + c_1,
\end{aligned}
\end{equation}
where $c_1, c_2, c_3 \in \mathbb{R}$ are the constants of integration determined by the initial conditions. This confirms that the method correctly integrates the system without requiring the transformation to a time-independent autonomous system or the search for a second commuting first integral in the extended phase space.
\end{example}

In this sense, time-dependent Hamiltonian systems provide a natural
source of exactly solvable dynamics that can be integrated without
appealing to Liouville--Arnold or Mishchenko--Fomenko integrability,\linebreak[3] and
are naturally captured by Poisson \cinf-structures.

\subsection{Applications to Vlasov Plasma}

The Vlasov equation provides a fundamental kinetic description of collisionless
plasmas and other many-particle systems. We consider the one-dimensional Vlasov
equation for a distribution function $f(x,v,t)$,
\begin{equation}\label{eq:vlasov}
\partial_t f + v\,\partial_x f + F[f]\,\partial_v f = 0,
\end{equation}
where $x,v\in\mathbb R$ and $t\ge0$. Equation \eqref{eq:vlasov} expresses the
conservation of $f$ along the phase-space characteristics
\[
\dot x = v,
\qquad
\dot v = F[f].
\]
The force term $F[f]$ depends functionally on the distribution function $f$,
rendering \eqref{eq:vlasov} a nonlinear integro-differential equation. Important
examples include the Vlasov--Poisson and Vlasov--Maxwell systems.

The Vlasov equation admits a noncanonical Hamiltonian formulation on the space
of distribution functions. The Hamiltonian functional is
\begin{equation}\label{hamiltonian}
\mathcal{H}[f]
=
\int \tfrac12 v^2 f(x,v)\,dx\,dv + \mathcal{V}[f],
\end{equation}
where $\mathcal V[f]$ denotes interaction or field energy. The evolution can
be written as
\[
\partial_t f = \{f,\mathcal{H}\},
\]
with respect to the Vlasov Poisson bracket \eqref{eq:vlasov-bracket}. Velocity moments
\begin{equation}\label{momentsintegral}
m_n(x,t)=\int v^n f(x,v,t)\,dv,\qquad n\ge0,
\end{equation}
define a mapping
\[
\mathcal M:\mathfrak g^* \to \mathbb R^{N+1},\qquad
f\mapsto(m_0,\dots,m_N).
\]
This mapping does not define an invariant submanifold of the Vlasov dynamics,
since the evolution of $m_n$ depends on higher-order moments. Moreover, the
Vlasov Lie--Poisson bracket does not generally close on functions of the
retained moments. Burby’s approach \cite{Burby2023VariableMoment} addresses this algebraically by postulating closure relations
\[
m_{N+k}=\Phi_{N+k}(m_0,\dots,m_N),\qquad k\ge1,
\]
and requiring that the induced bracket satisfy the Jacobi identity. When this
condition holds, the reduced system defines a finite-dimensional
Lie--Poisson structure \cite{ScovelWeinstein1994}, although the reduction is not dynamically invariant.
Hamiltonian consistency is preserved, but exact correspondence with Vlasov
solutions is not guaranteed.

This motivates the search for additional geometric constraints beyond the
Jacobi identity. Rather than focusing on commuting integrals, we study the
internal organization of the Poisson bracket itself. If the bracket closes in
a triangular manner on an ordered family of functions, the Hamiltonian vector
field generates an integrable one-dimensional distribution.

Poisson \cinf-structures therefore act as a selection principle within the
space of Hamiltonian reductions, isolating special subclasses for which the
Poisson algebra is sufficiently rigid to allow explicit integration. 
We therefore restrict attention to \emph{waterbag distributions}.

\subsubsection{Waterbags}

We consider a purely kinetic Vlasov model and therefore set the interaction
potential to zero. We restrict attention to spatially homogeneous distribution
functions $f=f(v)$. Under this assumption, the Hamiltonian functional reduces to
\begin{equation}\label{hamwaterbag1}
\mathcal{H}[f]
=
\int \frac12 v^2 f(v)\,dv.
\end{equation}

We consider waterbag distribution functions of the form
\begin{equation}\label{eq:waterbag}
f(v)
=
\sum_{k=1}^N
\chi_{[v_k^-,\,v_k^+]}(v),
\end{equation}
where $\chi_{[a,b]}$ denotes the characteristic function of the interval
$[a,b]$, equal to $1$ for $a\le v\le b$ and zero otherwise. The parameters
$v_k^\pm$ represent the lower and upper velocity boundaries of each bag.
The family \eqref{eq:waterbag} defines a finite-dimensional submanifold of the
space of distribution functions, parametrized by the $2N$ variables
$(v_1^-,v_1^+,\dots,v_N^-,v_N^+)$.

Substituting \eqref{eq:waterbag} into the Hamiltonian functional and integrating
over $v$ yields a finite-dimensional Hamiltonian function
\begin{equation}\label{Ham2}
H(v_1^\pm,\dots,v_N^\pm)
=
\sum_{k=1}^N
\frac16\bigl[(v_k^+)^3-(v_k^-)^3\bigr],
\end{equation}
which is an exact restriction of the Vlasov kinetic energy to the manifold of
waterbag distributions.

\medskip
\noindent
\textbf{Reduction of the Vlasov Lie--Poisson bracket.}
We now determine the Poisson structure induced on the waterbag parameters
$v_k^\pm$ by restriction of the Vlasov Lie--Poisson bracket. At the
infinite-dimensional level, the Vlasov equation admits the Lie--Poisson
formulation
\begin{equation}\label{eq:vlasov-bracket}
\{F,G\}(f)
=
\int_{\mathbb R^2}
f(x,v)\,
\Bigl\{
\frac{\delta F}{\delta f},
\frac{\delta G}{\delta f}
\Bigr\}_{\mathrm{can}}
\,dx\,dv,
\end{equation}
where the canonical Poisson bracket on phase space is $\{a,b\}_{\mathrm{can}}
=
\partial_x a\,\partial_v b - \partial_v a\,\partial_x b$. Since we work throughout with spatially homogeneous distributions $f=f(v)$,
the dependence on $x$ is trivial and $\partial_x=0$. Consequently, the Vlasov
bracket reduces to boundary contributions associated with the discontinuities
of $f$ in velocity space. Let
\[
F(f)=\tilde F(v_1^\pm,\dots,v_N^\pm),
\qquad
G(f)=\tilde G(v_1^\pm,\dots,v_N^\pm)
\]
be functionals depending on $f$ only through the waterbag boundary variables.

To compute variations of $f$ with respect to the parameters $v_k^\pm$, it is
convenient to write the characteristic function as $\chi_{[a,b]}(v)=H(v-a)-H(v-b),$
where $H$ denotes the Heaviside function. Differentiating with respect to the
parameters $a$ and $b$ yields

\[
\delta f(v)
=
\sum_{k=1}^N
\left[
\delta(v-v_k^+)\,\delta v_k^+
-
\delta(v-v_k^-)\,\delta v_k^-
\right].
\]
So that by definition of the functional derivative, $\delta F
=
\int \frac{\delta F}{\delta f}(v)\,\delta f(v)\,dv,$
\[
\delta F
=
\sum_{k=1}^N
\left(
\frac{\partial \tilde F}{\partial v_k^+}\,\delta v_k^+
+
\frac{\partial \tilde F}{\partial v_k^-}\,\delta v_k^-
\right).
\]
Identifying coefficients, we obtain
\begin{equation}\label{eq:functional-derivatives}
\frac{\delta F}{\delta f}(v_k^+)
=
\frac{\partial \tilde F}{\partial v_k^+},
\qquad
\frac{\delta F}{\delta f}(v_k^-)
=
-\frac{\partial \tilde F}{\partial v_k^-}.
\end{equation}

Since $f$ is piecewise constant in $v$, the only contributions to the bracket
\eqref{eq:vlasov-bracket} arise at the discontinuities of $f$. Using
\eqref{eq:functional-derivatives}, the reduced bracket on the waterbag
parameters is
\begin{equation}\label{eq:reduced-bracket}
\{F,G\}
=
\sum_{k=1}^N
\left(
\frac{\partial \tilde F}{\partial v_k^+}
\frac{\partial \tilde G}{\partial v_k^-}
-
\frac{\partial \tilde F}{\partial v_k^-}
\frac{\partial \tilde G}{\partial v_k^+}
\right).
\end{equation}

Applying \eqref{eq:reduced-bracket} to the coordinate functions
$v_i^\pm$, we obtain
\[
\{v_i^+,v_j^-\}=\delta_{ij},
\qquad
\{v_i^\pm,v_j^\pm\}=0,
\]
with all brackets between different waterbags vanishing. This defines a
finite-dimensional Poisson structure which is the exact restriction of the
Vlasov Lie--Poisson bracket to the manifold of waterbag distributions.
It is convenient to introduce the width and center variables
\[
w_k := v_k^+ - v_k^-,
\qquad
c_k := v_k^+ + v_k^-,
\]
so that $v_k^\pm=\tfrac12(c_k\pm w_k)$. In these variables the Poisson bracket
takes the canonical form
\[
\{w_k,c_\ell\}=2\,\delta_{k\ell},
\qquad
\{w_k,w_\ell\}=\{c_k,c_\ell\}=0,
\]
and the Hamiltonian \eqref{Ham2} becomes
\[
H
=
\sum_{k=1}^N
\left(
\frac14\,c_k^2 w_k
+
\frac{1}{12}\,w_k^3
\right),
\]
a cubic polynomial completely decoupled by bags. 

We will illustrate the Pfaffian integration procedure for the decoupled
waterbag dynamics by considering the simplest nontrivial case
$N=2$. The phase space has dimension $4$, with coordinates
$(w_1,c_1,w_2,c_2)$, and the Hamiltonian reads
\[
H = H_1 + H_2,
\qquad
H_k=\frac14\,c_k^2 w_k+\frac{1}{12}\,w_k^3.
\]

A direct calculation from the Hamiltonian for two waterbags yields:
\[
X_H
=
c_1 w_1\,\partial_{w_1}
-\frac12(c_1^2+w_1^2)\partial_{c_1}
+
c_2 w_2\,\partial_{w_2}
-\frac12(c_2^2+w_2^2)\partial_{c_2}.
\]
The Hamiltonian distribution to be integrated is therefore $\mathcal D_H=\langle X_H\rangle $,
which has rank $k=1$. The associated symplectic form is
\[
\omega=\frac12\bigl(dw_1\wedge dc_1+dw_2\wedge dc_2\bigr),
\]
and the Liouville volume form is
\[
\Omega=\frac14\,dw_1\wedge dc_1\wedge dw_2\wedge dc_2.
\]

Since the dimension is $2n=4$, a Poisson $\mathcal C^\infty$–structure
consists of $2n-2=2$ functions. We define the ordered family
\[
\mathcal F=(f_1,f_2):=(H_1,w_1),
\]
and set $f_0:=H$.
The functions $(f_0,f_1,f_2)$ are functionally independent on the open
set $c_1w_1\neq0$.

The Poisson brackets among the functions $(f_0,f_1,f_2)$ satisfy:
\begin{align*}
\{f_1,f_0\}&=F_{10}=\{H_1,H_1+H_2\}=0,\\
\{f_2,f_0\}&=F_{20}=\{w_1,H_1+H_2\}=\{w_1,H_1\}=c_1 w_1,\\
\{f_2,f_1\}&=F_{21}=\{w_1,H_1\}=c_1 w_1.
\end{align*}
Taking into account that $c_1^2=\frac{4f_1}{f_2}-\frac13 f_2^2$, the above brackets satisfy condition \eqref{condicion}, so that $\mathcal F$ defines a Poisson \cinf-structure for the system. 

The remaining relevant brackets are
\begin{align*}
\{f_3,f_0\}&=F_{30}=\{w_2,H_1+H_2\}=\{w_2,H_2\}=c_2 w_2,\\
\{f_3,f_1\}&=F_{31}=\{w_2,H_1\}=0,\\
\{f_3,f_2\}&=F_{32}=\{w_2,w_1\}=0.
\end{align*}

Following the algorithm in Section~\ref{sec:integration-algorithm}, we compute the three Pfaffian 1-forms using the explicit $n=2$ formulas in Example~\ref{casen2}:
\begin{align*}
\omega_1 &= F_{23}\,df_0 - F_{03}\,df_2 + F_{02}\,df_3
= (c_2w_2)\,dw_1 - (c_1w_1)\,dw_2,\\
\omega_2 &= F_{13}\,df_0 - F_{03}\,df_1 + F_{01}\,df_3
= (c_2w_2)\,dH_1,\\
\omega_3 &= F_{12}\,df_0 - F_{02}\,df_1 + F_{01}\,df_2
= -(c_1w_1)\,d(H-H_1)
=-(c_1w_1)\,dH_2.
\end{align*}

We now solve the Pfaffian equations sequentially. On any open neighborhood where $c_1w_1\neq 0$, the equation $\omega_3\equiv 0$ implies $dH_2\equiv 0$. Hence, we obtain the first integral (energy of the second subsystem)
\[
H_2=\frac14c_2^2w_2+\frac{1}{12}w_2^3.
\]
Restricting to the level set with $H_2$ constant, and solving for $c_2$, we obtain
\begin{equation}\label{rel_c2}
c_2 = \pm\sqrt{\frac{4H_2}{w_2} - \frac{w_2^2}{3}}.
\end{equation}

On this leaf, the equation $\omega_2\equiv 0$ reduces (for $c_2w_2\neq 0$) to $dH_1\equiv 0$, and we obtain
\[
H_1=\frac14c_1^2w_1+\frac{1}{12}w_1^3.
\]

Solving for $c_1$ on a level set with $H_1$ constant gives
\begin{equation}\label{rel_c1}
c_1 = \pm\sqrt{\frac{4H_1}{w_1} - \frac{w_1^2}{3}}.
\end{equation}

Finally, restricting $\omega_1$ to the level set above, and using
\[
c_iw_i=\pm\sqrt{4H_iw_i-\tfrac13 w_i^4}=\pm\frac{1}{\sqrt{3}}\sqrt{12H_iw_i-w_i^4},\qquad i=1,2,
\]
we obtain
\[
\frac{dw_1}{\sqrt{4H_1 w_1 - \tfrac13 w_1^4}} - \frac{dw_2}{\sqrt{4H_2 w_2 - \tfrac13 w_2^4}} = 0,
\]
where the overall sign is fixed locally by continuity along integral curves.
This implies that both terms equal a common differential $dt$, yielding the separated equations of motion
\begin{equation}\label{eq:w1motion}
\frac{dw_1}{\sqrt{4H_1 w_1 - \tfrac13 w_1^4}} = dt, \qquad
\frac{dw_2}{\sqrt{4H_2 w_2 - \tfrac13 w_2^4}} = dt.
\end{equation}

The explicit integration of \eqref{eq:w1motion} yields the solution in terms of the Weierstrass elliptic function $\wp(z;g_2,g_3)$. For the $i$-th waterbag ($i=1,2$), the width evolves as
\begin{equation}\label{eq:weierstrass_sol}
    w_i(t) = \frac{H_i}{\wp(t + \tau_i; 0, H_i^2/3)},
\end{equation}
where $\tau_i$ is an integration constant determined by the initial conditions. In this equianharmonic case,
\[
g_2 = 0, \qquad g_3 = \frac{H_i^2}{3},
\]
so the period lattice is triangular (hexagonal) and $w_i(t)$ is periodic, oscillating between $w=0$ and the maximal width $w_{\max}=\sqrt[3]{12H_i}$.

Finally, the center variable is recovered as
\[
c_i(t)=\pm\sqrt{\frac{4H_i}{w_i(t)}-\frac{w_i(t)^2}{3}},
\]
where the sign is fixed by the initial condition (and hence locally by continuity along the integral curve).

\section{Generalization of Poisson \texorpdfstring{$C^\infty$}{Cinfty} Structures to Jacobi Manifolds: Jacobi \texorpdfstring{$C^\infty$}{Cinfty} Structures}\label{sec:jacobicinf}
In this section, we generalize the previously investigated Poisson $C^\infty$-structures to Jacobi manifolds, resulting in Jacobi $C^\infty$ structures. The Jacobi manifolds involved specifically include Poisson, symplectic, locally conformally symplectic, and contact manifolds as subclasses of Jacobi manifolds, and we illustrate each with examples \cite{sardon1,sardon2,sardon3}. The framework relies on the existence of a Lie bracket on functions and a mapping $f \mapsto X_f$ of Hamiltonian vector fields satisfying
\[
[X_f,X_g] = X_{\{f,g\}}.
\]
This property links algebraic relations among functions to closure properties of Hamiltonian vector fields, independent of the Leibniz rule or non-degeneracy.

\begin{definition}\label{def:jac}
A \emph{Jacobi manifold} is a triple $(M,\Lambda,E)$, where $M$ is a smooth
manifold, $\Lambda$ is a bivector field and $E$ is a vector field on $M$,
satisfying
\[
[\Lambda,E]_{SN}=0,
\qquad
[\Lambda,\Lambda]_{SN}=2E\wedge\Lambda,
\]
where $[\cdot,\cdot]_{SN}$ denotes the Schouten--Nijenhuis bracket. The associated
Jacobi bracket on $C^\infty(M)$ is defined by
\[
\{f,g\}=\Lambda(df,dg)+fE(g)-gE(f).
\]
Given a function $f\in C^\infty(M)$, the corresponding Hamiltonian vector field is
\begin{align}\label{Ha} X_f=\Lambda^\sharp(df)+fE. 
\end{align}
\end{definition}
Classical geometric structures appearing in Hamiltonian dynamics arise as
particular realizations of Jacobi manifolds. Their mutual relations are summarized in the following paragraph for Jacobi substructures.
\begin{center}
\begin{tikzpicture}[
  level distance=1.2cm,
  sibling distance=4cm,
  every node/.style={align=center, font=\small}
]
\node {Jacobi geometry\\ $(M,\Lambda,E)$}
  child { node {Poisson geometry\\ $(M,\Lambda)$}
    child { node {Symplectic geometry\\ $(M,\omega)$} }
  }
  child { node {LCS geometry\\ $(M,\Omega,\theta)$} }
  child { node {Contact geometry\\ $(M,\eta)$} };
\end{tikzpicture}
\end{center}

\paragraph{Realizations of Jacobi manifolds.}
Let $(M,\Lambda,E)$ be a Jacobi manifold. The main geometric structures appearing
in Hamiltonian dynamics are distinguished by the properties of the vector field
$E$ and by dimensional considerations.

If $E=0$, the Jacobi identities reduce to $[\Lambda,\Lambda]=0$, and the Jacobi
bracket becomes a Poisson bracket. In this case, $(M,\Lambda)$ is a Poisson
manifold, and if $\Lambda$ is non-degenerate, it corresponds to a symplectic
manifold $(M,\omega)$.

If $E\neq 0$ and $\dim M=2n$, Jacobi manifolds give rise to locally conformally
symplectic structures. In this case, the bivector $\Lambda$ is induced by a
non-degenerate two-form $\Omega$ satisfying $d\Omega=\theta\wedge\Omega$, where $\theta$ is a closed 1-form and the
vector field $E$ coincides with the Lee vector field $Z_\theta$.

If $\dim M=2n+1$, Jacobi manifolds correspond to contact manifolds. The vector
field $E$ is (up to sign) the Reeb vector field associated with the contact
one-form $\eta$, and the Jacobi bracket coincides with the contact bracket.

\begin{table}[ht]
\centering
\footnotesize
\begin{tabular}{l l l l}
\hline
Structure & Geometric data & Jacobi vector $E$ & Dimension \\
\hline \\[0.1ex]
Jacobi
& $(M,\Lambda,E)$
& arbitrary
& any \\[0.6ex]

Poisson
& $(M,\Lambda)$
& $E=0$
& any \\[0.6ex]

Symplectic
& $(M,\omega)$
& $E=0$
& $2n$ \\[0.6ex]

LCS
& $(M,\Omega,\theta)$
& $E=Z_\theta=\Omega^\sharp(\theta)$
& $2n$ \\[0.6ex]

Contact
& $(M,\eta)$
& $E=-\mathcal{R}$
& $2n+1$ \\[0.6ex]
\\[0.1ex] \hline
\end{tabular}
\caption{{\small Jacobi geometry and its main particular cases.}}
\label{Jacobi-summary}
\end{table}

\paragraph{Hamiltonian dynamics on Jacobi realizations.}
Let $f\in C^\infty(M)$ be a Hamiltonian function. In a Jacobi manifold $(M, \Lambda, E)$, as given in equation \eqref{Ha}, the associated Hamiltonian vector field is defined as
\[
X_f = \Lambda^\sharp(df) + fE .
\]
This expression specializes in each of the above cases.

On a Poisson manifold, where $E=0$, one recovers the usual Hamiltonian vector
field
\[
X_f=\Lambda^\sharp(df).
\]

On a locally conformally symplectic manifold $(M,\Omega,\theta)$, the Hamiltonian
vector field is given by
\[
X_f\contract\Omega=d_\theta f,
\qquad
X_f=\Omega^\sharp(df)+fZ_\theta,
\]
where $d_\theta f=df-f\theta$ and  $Z_\theta$ is the vector field satisfying $Z_\theta\contract\Omega=-\theta.$

On a contact manifold $(M,\eta)$ with Reeb vector field $\mathcal{R}$, the Hamiltonian vector field $X_f$ associated to a smooth function $f$ is defined by the equations
\[
X_f\contract\eta = f, \qquad X_f\contract d\eta = df - \mathcal{R}(f)\eta,
\]
or equivalently,
\[
X_f = \Lambda^{\sharp}(df) - f\mathcal{R},
\]
where $\Lambda^\sharp: T^*M \to TM$ is the isomorphism induced by the restriction of $d\eta$ to the contact distribution $\ker\eta$. 

In all cases, the correspondence $f\mapsto X_f$ satisfies
\[
[X_f,X_g]=X_{\{f,g\}},
\]
which is the fundamental property exploited in the remainder of this work.

\subsection{Jacobi Hamiltonian systems}

Reconsider a $m$-dimensional Jacobi manifold $(M,\Lambda,E)$ as in Definition \ref{def:jac}.
From the previous section, we can state that unlike symplectic manifolds, a Jacobi manifold need not be
even-dimensional, and the Hamiltonian flow associated with a function
$H\in C^\infty(M)$ does not necessarily preserve $H$ itself. This phenomenon
occurs, for instance, in contact and locally conformally symplectic
manifolds.

Moreover, if the Jacobi vector field $E$ does not vanish, the Hamiltonian
vector field associated with the constant function $1$ satisfies
\[
X_1 = E \neq 0.
\]
This behavior has no analogue in symplectic or Poisson geometry, where
constant functions always generate the zero Hamiltonian vector field.
This motivates the introduction of Jacobi $C^\infty$-structures, which is more general than Poisson $C^\infty$-structure studied previously.

But first, we introduce the following result:

\begin{proposition}\label{JacC}
Let $(M,\Lambda,E)$ be a Jacobi manifold and let $f_0,\dots,f_j\in C^\infty(M)$. For any smooth function $F\in C^\infty(\mathbb{R}^{j+1})$ one has
\[
X_{F(f_0,\dots,f_j)}
=\sum_{k=0}^j \frac{\partial F}{\partial u_k}(f_0,\dots,f_j)\,X_{f_k}
+\Bigl(F(f_0,\dots,f_j)-\sum_{k=0}^j f_k\,\frac{\partial F}{\partial u_k}(f_0,\dots,f_j)\Bigr)E.
\]
In particular, if $\{f_j,f_i\}=F_{ji}(f_0,\dots,f_j)$ for some $F_{ji}$, then
\[
[X_{f_j},X_{f_i}]=X_{\{f_j,f_i\}}\in \langle X_{f_0},\dots,X_{f_j},E\rangle.
\]
\end{proposition}
\begin{proof}
By definition, $X_h=\Lambda^\sharp(dh)+hE$ for all $h\in C^\infty(M)$. Using the chain rule,
\[
d\bigl(F(f_0,\dots,f_j)\bigr)=\sum_{k=0}^j \frac{\partial F}{\partial u_k}(f_0,\dots,f_j)\,df_k.
\]
Hence
\begin{align*}
X_{F(f_0,\dots,f_j)}
&=\Lambda^\sharp\!\left(\sum_{k=0}^j \frac{\partial F}{\partial u_k}(f_0,\dots,f_j)\,df_k\right)+F(f_0,\dots,f_j)E \\
&=\sum_{k=0}^j \frac{\partial F}{\partial u_k}(f_0,\dots,f_j)\,\Lambda^\sharp(df_k)+F(f_0,\dots,f_j)E \\
&=\sum_{k=0}^j \frac{\partial F}{\partial u_k}(f_0,\dots,f_j)\,(X_{f_k}-f_kE)+F(f_0,\dots,f_j)E,
\end{align*}
which is the desired identity.
\end{proof}

Now we are in a position to introduce the main definition of this section, which generalizes the notion of Poisson $C^\infty$-structure to Jacobi manifolds:
\begin{definition}\label{jacobicinf}
Let $(M,\Lambda,E)$ be a Jacobi manifold of dimension $m$ and let $H\in C^\infty(M)$. An ordered family of $m-2$ functions
\[
\mathcal F=(f_1,\dots,f_{m-2})
\]
is called a \emph{Jacobi $C^\infty$-structure} for the Jacobi Hamiltonian system $(M,H)$ if, taking $f_0=H$, the following conditions hold:

{
\renewcommand{\labelenumi}{\textup{(\arabic{enumi})}}
\begin{enumerate}
\item \textup{Vector field independence.} The Hamiltonian vector fields
\[
X_{f_0},X_{f_1},\dots,X_{f_{m-2}}
\]
are pointwise linearly independent on an open dense subset of $M$.

\item \textup{Triangular bracket closure.} For every pair $j>i\geq 0,$ there exists a smooth function
\(
F_{ji}=F_{ji}(u_0,\cdots,u_j)\in C^\infty(\mathbb{R}^{j+1})
\)
such that
\[
\{f_j,f_i\}=F_{ji}(f_0,f_1,\dots,f_j).
\]

\item \textup{Reeb compatibility.} For every $j>i\geq 0$, the vector field 
\[
\mathcal E_{ji}:=\Bigl(F_{ji}(f_0,\dots,f_j)-\sum_{k=0}^j f_k\,\frac{\partial F_{ji}}{\partial u_k}(f_0,\dots,f_j)\Bigr)\,E
\]
belongs (locally) to $\langle X_{f_0},X_{f_1},\dots,X_{f_j}\rangle$.
\end{enumerate}
}
\end{definition}

\begin{remark}\label{aboutC3}
Condition \textup{(3)} reflects a fundamental difference between Poisson and Jacobi geometry. 
In the Jacobi case, the map $f \mapsto X_f$ is \emph{affine} rather than linear, as can be seen in equation \eqref{Ha}. So condition \textup{(3)} is required to guarantee that the Lie brackets of the Hamiltonian vector fields close in a triangular way. Indeed, according to Proposition~\ref{JacC}, if condition \textup{(2)} holds, then
\[
[X_{f_j},X_{f_i}]=X_{\{f_j,f_i\}}=X_{F_{ji}}\in \langle X_{f_0},\dots,X_{f_j}\rangle
\]
whenever $\mathcal E_{ji}$ lies in the same submodule. 
\end{remark}

\begin{remark}
\textup{(Special cases).}
\begin{itemize}
\item If $E=0$ (Poisson/symplectic case), then $\mathcal E_{ji}=0$ and condition \textup{(3)} is automatically satisfied. In particular, for nondegenerate $\Lambda$ one recovers the Poisson \cinf-structures previously introduced on symplectic manifolds.

\item Assume $E\neq 0$ and $f_1=1$. Since $X_{f_1}=X_1=E$, the vector field $\mathcal E_{ji}$ (which is always proportional to $E$) automatically belongs to $\langle X_{f_0},\dots,X_{f_j}\rangle$ for every $j>i$. Hence, condition \textup{(3)} becomes automatic in this case. Nevertheless, since 
$$
\{f_1,f_0\}=\{1,H\}=E(H),
$$
the closure condition \textup{(2)} implies that $E(H)$ must be a function of $H$ alone, which is a nontrivial constraint on the Hamiltonian.

\item More generally, if for each pair $j>i$ the closure function $F_{ji}=F(u_0,\dots,u_j)$ is $1$-homogeneous in $(u_0,\dots,u_j)$, then Euler's theorem implies
\[
F_{ji}=\sum_{k=0}^j u_k\,\frac{\partial F_{ji}}{\partial u_k},
\]
and hence $\mathcal E_{ji}=0$ for that bracket relation. 
\end{itemize}
\end{remark}

Based on the preceding discussion, we obtain the following theorem.

\begin{theorem}[Exact integration on Jacobi manifolds]\label{thm:jacobi-main}
Let $(M,H)$ be a Jacobi Hamiltonian system admitting a Jacobi
$C^\infty$-structure $\mathcal F=(f_1,\dots,f_{m-2})$. Then the Hamiltonian distribution
\[
\mathcal D_H=\langle X_H\rangle
\]
admits a $C^\infty$-structure and the equations of motion can be integrated
explicitly by solving a
sequence of $m-1$ completely integrable Pfaffian equations.
\end{theorem}

\begin{proof}
Set $f_0:=H$, so that $\mathcal{D}_H=\langle X_{f_0}\rangle$.
By condition \textup{(1)} in Definition~\ref{jacobicinf}, the vector fields
\[
X_{f_0},X_{f_1},\dots,X_{f_{m-2}}
\]
are pointwise linearly independent on an open dense subset $U\subset M$.
On $U$, choose any vector field $R$ pointwise linearly independent of them.

We claim that the ordered family
\[
\bigl(X_{f_1},\dots,X_{f_{m-2}},R\bigr)
\]
is a \cinf-structure for $\mathcal{D}_H$ in the sense of Definition~\ref{def1}.
For $1\le j\le m-2$, let
\[
\mathcal{D}_{j-1}:=\mathcal{D}_H\oplus\langle X_{f_1},\dots,X_{f_{j-1}}\rangle
=\langle X_{f_0},X_{f_1},\dots,X_{f_{j-1}}\rangle.
\]
We show that $X_{f_j}$ is a \cinf-symmetry of $\mathcal{D}_{j-1}$.
It suffices to check Lie brackets with the generators $X_{f_i}$, $0\le i\le j-1$.
For such a pair $j>i$, condition \textup{(2)} in Definition~\ref{jacobicinf} gives
\(
\{f_j,f_i\}=F_{ji}(f_0,\dots,f_j)
\)
for some $F_{ji}\in C^\infty(\mathbb{R}^{j+1})$.
Using the basic identity $[X_f,X_g]=X_{\{f,g\}}$ recalled at the beginning of Section~\ref{sec:jacobicinf}, we obtain
\[
[X_{f_j},X_{f_i}]=X_{\{f_j,f_i\}}=X_{F_{ji}(f_0,\dots,f_j)}.
\]
Applying Proposition~\ref{JacC} to $F_{ji}$ yields
\[
X_{F_{ji}(f_0,\dots,f_j)}=\sum_{k=0}^j \frac{\partial F_{ji}}{\partial u_k}(f_0,\dots,f_j)\,X_{f_k}+\mathcal{E}_{ji},
\]
where $\mathcal{E}_{ji}$ is exactly the Reeb term defined in condition \textup{(3)} of Definition~\ref{jacobicinf}.
By that condition, $\mathcal{E}_{ji}\in\langle X_{f_0},\dots,X_{f_j}\rangle$ (locally on $U$), hence
\[
[X_{f_j},X_{f_i}]\in\langle X_{f_0},\dots,X_{f_j}\rangle\subset \mathcal{D}_{j-1}\oplus\langle X_{f_j}\rangle.
\]
Therefore $X_{f_j}$ is a \cinf-symmetry of $\mathcal{D}_{j-1}$ for each $1\le j\le m-2$.

Finally, by Remark~\ref{rem:last}, the last vector field of a \cinf-structure may be taken to be any vector field pointwise linearly independent of the previous ones and of the generators of $\mathcal{D}_H$; thus the above choice of $R$ completes a \cinf-structure for $\mathcal{D}_H$ on $U$.
Applying Theorem~\ref{TZ} yields the explicit (local) integration of the equations of motion by solving a sequence of $m-1$ completely integrable Pfaffian equations.
\end{proof}

Below, we provide a detailed discussion of several specific cases of Jacobi manifolds: Poisson manifolds, LCS manifolds, and contact manifolds, and obtain results analogous to Theorem \ref{thm:jacobi-main} for each case.

\subsection{Poisson Hamiltonian systems}

Let $(M,\Lambda)$ be a Poisson manifold of dimension $m$. This corresponds to the
Jacobi structure $(\Lambda,E)=(\Lambda,0)$.
\begin{theorem}\label{Pio}
Let $(M,\Lambda)$ be a $m$-dimensional Poisson manifold, and\linebreak[4] let
$\mathcal{F}=(f_1,\dots,f_{m-2})$ be a Jacobi $C^\infty$-structure of the Hamiltonian system $(M,H)$.
 Then the Hamiltonian distribution
\[
\mathcal{D}_H=\langle X_H\rangle
\]
admits a \cinf-structure generated by
\[
\{X_{f_1},X_{f_2},\dots,X_{f_{m-2}},R\},
\]
where $R$ is an arbitrary vector field, pointwise linearly independent of $\{X_{f_i}\}$.

Consequently, the equations of motion can be integrated explicitly by solving a
sequence of $m-1$ completely integrable Pfaffian equations.
\end{theorem}
\begin{proof}
  The proof is similar to that of Theorem \ref{thm:main}.
\end{proof}

\begin{example}[Jacobi $C^\infty$-structure on a three-dimensional Poisson manifold]
Consider the manifold \(M = \mathbb{R}^3\) with coordinates \((x, y, z)\).  
Define the Poisson tensor
\[
\Lambda = \frac{\partial}{\partial x} \wedge \frac{\partial}{\partial y}.
\]
Take the Hamiltonian \(H = x + y\) and let  
\(f_1 = e^{\frac{x^2 - y^2}{2}}\).  
A direct computation gives the Hamiltonian vector fields
\[
X_H = -\frac{\partial}{\partial x} +\frac{\partial}{\partial y}, 
\qquad
X_{f_1} =ye^{\frac{x^2-y^2}{2}}\frac{\partial}{\partial x} +xe^{\frac{x^2-y^2}{2}}\frac{\partial}{\partial y},
\]
and their Poisson bracket
\[
\{H, f_1\} = X_H(f_1) = -e^{\frac{x^2-y^2}{2}}(x+y) = -f_1\cdot H.
\]
Hence \(\mathcal{F} = (f_1)\) satisfies the triangular closure condition and constitutes a Jacobi \(C^\infty\)-structure of $(M,H)$.  

Now choose the vector field \(R := \dfrac{\partial}{\partial z}\).  
The Hamiltonian distribution \(\mathcal{D}_H = \langle X_H \rangle\) then admits a \(C^\infty\)-structure generated by
\[
\{X_{f_1},\, R\} = \Big\{ye^{\frac{x^2-y^2}{2}}\frac{\partial}{\partial x} +xe^{\frac{x^2-y^2}{2}}\frac{\partial}{\partial y},\; \frac{\partial}{\partial z}\Big\}.
\]
By Theorem~\ref{Pio}, the equations of motion defined by \(X_H\) can therefore be integrated explicitly. In fact, taking a volume form $\Omega=dx\wedge dy\wedge dz$, then we can calculate that
\begin{align*}
\omega_2:=X_{f_1}\contract X_H\contract\Omega=-e^{\frac{x^2-y^2}{2}}(x+y)dz,\quad \omega_1=R\contract X_H\contract\Omega=dy+dx.
\end{align*}
It is clear that 
$$dz\wedge\omega_2=0,\quad X_H\contract\omega_2=0$$
thus that $I_1=z$ is a first integral of $X_H.$ Moreover, the restriction of the 1-form $\omega_1$ to the set $\Sigma_c=\{(x,y,z)\in\mathbb R^3\;|\;z=c\}$ still has the form $\omega_1|_{\Sigma_c}=dy+dx$. Let $I_2=x+y,$ we know that 
$$dI_2\wedge\omega_1|_{\Sigma_c}=0, $$
which means that $I_2$ is also a first integral of $X_H.$ So, $X_H$ is completely integrable in $\mathbb R^3.$

\end{example}

\subsection{Locally conformally symplectic systems}

Let $(M,\omega,\theta)$ be a locally conformally symplectic manifold of dimension
$2n$, with $d\omega=\theta\wedge\omega$ \cite{Zhao5,Azuaje}. The inverse bivector $\Lambda=\omega^{-1}$
defines a Jacobi structure $(\Lambda,E)$ with
\[
E=\Lambda^\sharp(\theta).
\]
Hamiltonian vector fields take the form
\[
X_f=\Lambda^\sharp(df)+fE.
\]

According to Theorem~\ref{thm:jacobi-main}, we can state the following result for Hamiltonian systems on LCS manifolds:
\begin{theorem}\label{LCSC}
A Hamiltonian system on a 2$n$-dimensional locally conformally symplectic manifold that admits a Jacobi $C^\infty$-structure is exactly integrable by solving a
sequence of $2n-1$ completely integrable Pfaffian equations, requiring only $2n-2$ additional functions beyond the Hamiltonian itself.
\end{theorem}

\begin{example}[Jacobi $C^\infty$-structure on a 4-dimensional LCS manifold]
\label{ex:LCS-Jacobi-Cinf}

Consider the manifold $M = \mathbb{R}^4$ with coordinates $(x, y, w, z)$.
Define the 2-form and the closed 1-form
\[
\omega = e^{x}\big( dx \wedge dy + dw \wedge dz \big), \qquad 
\theta = dx .
\]
A direct computation shows $d\omega = \theta \wedge \omega$; hence $(M,\omega,\theta)$ is a locally conformally symplectic (LCS) manifold of dimension $4$.

\medskip\noindent
\textbf{Associated Jacobi structure.}
The bivector field $\Lambda = \omega^{-1}$ is
\[
\Lambda = e^{-x}\big( \partial_x \wedge \partial_y + \partial_w \wedge \partial_z \big),
\]
and the Reeb vector field equals
\[
E = \Lambda^\sharp(\theta) = e^{-x}\,\frac{\partial}{\partial y}.
\]

For any two smooth functions $f,g$ the Jacobi bracket reads
\begin{equation}\label{eq:Jacobi-bracket-ex}
\{f,g\}=e^{-x}\!\Bigl[
\frac{\partial f}{\partial x}\frac{\partial g}{\partial y}
-\frac{\partial f}{\partial y}\frac{\partial g}{\partial x}
+\frac{\partial f}{\partial w}\frac{\partial g}{\partial z}
-\frac{\partial f}{\partial z}\frac{\partial g}{\partial w}
+f\frac{\partial g}{\partial y}-g\frac{\partial f}{\partial y}\Bigr].
\end{equation}

\medskip\noindent
\textbf{Jacobi $C^\infty$-structure.}
According to the definition, for the Hamiltonian $H=e^xy$ we choose $m-2=2$ additional functions
\[
f_1 = 1,\qquad f_2 = z.
\]
Let $f_0=H,$ using \eqref{eq:Jacobi-bracket-ex} one obtains the non‑zero brackets
\begin{align}
\{f_1,f_0\} &= \{1,\, e^x y\} \nonumber \\
            &= e^{-x}\!\left( \frac{\partial 1}{\partial x}\frac{\partial (e^x y)}{\partial y}
            - \frac{\partial 1}{\partial y}\frac{\partial (e^x y)}{\partial x}
            + \frac{\partial 1}{\partial w}\frac{\partial (e^x y)}{\partial z}
            - \frac{\partial 1}{\partial z}\frac{\partial (e^x y)}{\partial w}  + 1\cdot\frac{\partial (e^x y)}{\partial y}
            - e^x y\cdot\frac{\partial 1}{\partial y} \right) \nonumber \\
            &= e^{-x}(e^x) = 1=f_1. \label{eq:bracket13calc}\\
\{f_2,f_0\} &= \{z,\, e^x y\} \nonumber \\
            &= e^{-x}\!\left( \frac{\partial z}{\partial x}\frac{\partial (e^x y)}{\partial y}
            - \frac{\partial z}{\partial y}\frac{\partial (e^x y)}{\partial x}
            + \frac{\partial z}{\partial w}\frac{\partial (e^x y)}{\partial z}
            - \frac{\partial z}{\partial z}\frac{\partial (e^x y)}{\partial w} + z\frac{\partial (e^x y)}{\partial y}
            - e^x y\frac{\partial z}{\partial y} \right) \nonumber \\
            &= e^{-x}\bigl(  z e^x\bigr) =  z=   f_2. \label{eq:bracket23calc}\end{align}
\begin{align}
            \{f_1,f_2\} &= \{1,\, z\} \nonumber \\
            &= e^{-x}\!\left( \frac{\partial 1}{\partial x}\frac{\partial z}{\partial y}
            - \frac{\partial 1}{\partial y}\frac{\partial z}{\partial x}
            + \frac{\partial 1}{\partial w}\frac{\partial z}{\partial z}
            - \frac{\partial 1}{\partial z}\frac{\partial z}{\partial w}+ 1\cdot\frac{\partial z}{\partial y}
            - z\cdot\frac{\partial 1}{\partial y} \right) = 0. \label{eq:bracket12calc}
\end{align}
Thus there exist constants $a_{ij}^{k}\in\mathbb{R}$ such that
$\{f_i,f_j\}=\sum_{k=0}^{2}a_{ij}^{k}f_k$; explicitly
\[
a_{10}^{1}=-a_{01}^1=1,\quad a_{20}^{2}=-a_{02}^2=1,
\]
and all remaining $a_{ij}^{k}=0$.

Hence $(f_1,f_2)$ is a genuine Jacobi $C^\infty$-structure on the LCS manifold $(M,\omega,\theta)$. Hence, by Theorem \ref{LCSC}, we know that $X_H$ can be integrated explicitly. In fact, we can calculate that
\begin{align*}
X_H&=2y\frac{\partial}{\partial y}-\frac{\partial}{\partial x},\quad X_{f_1}=E=e^{-x}\frac{\partial}{\partial y},\quad X_{f_2}=-e^{-x}\frac{\partial}{\partial w}+ze^{-x}\frac{\partial}{\partial y}.
\end{align*}
In fact, by using
$$X_H=\Lambda^\sharp(dH)+HE,$$ we know that
\begin{align*}
    X_H&=\big(e^{-x}( \partial_x \wedge \partial_y + \partial_w \wedge \partial_z )\big)\contract dH+HE\\
    &=e^{-x}(\partial_x(H)\partial_y-\partial_y(H)\partial_x+\partial_w(H)\partial_z-\partial_z(H)\partial_w)+He^{-x}\frac{\partial}{\partial y}\\
    &=2y\frac{\partial}{\partial y}-\frac{\partial}{\partial x},
\end{align*}
similarly, we can get
$$X_{f_1}=E=e^{-x}\frac{\partial}{\partial y},\quad X_{f_2}=-e^{-x}\frac{\partial}{\partial w}+ze^{-x}\frac{\partial}{\partial y}.$$

Taking another vector field $R=\frac{\partial}{\partial z}$ and a volume form $\Omega=dx\wedge dy\wedge dz\wedge dw,$ we can get that
\begin{align*}
\omega_3:&=X_{f_1}\contract X_{f_2}\contract X_H\contract\Omega=e^{-2x}dz, \\ \omega_2:&=R\contract X_{f_1}\contract X_H\contract\Omega=e^{-x}dw,\\
\omega_1:&=R\contract X_{f_2}\contract X_H\contract\Omega=e^{-x}dy+2ye^{-x}dx.
\end{align*}
It is clear that 
$$dz\wedge\omega_3=0,\quad X_H\contract\omega_3=0$$
thus that $I_1=z$ is a first integral of $X_H.$ Moreover, the restriction of the 1-form $\omega_2$ to the set $\Sigma_{c_1}=\{(x,y,w,z)\in\mathbb R^4\;|\;z=c_1\}$ still has the form $\omega_2|_{\Sigma_{c_1}}=e^{-x}dw$. Let $I_2=w,$ we know that 
$$dI_2\wedge\omega_2|_{\Sigma_{c_1}}=0, $$
which means that $I_2$ provides also a first integral of $X_H.$ Further, it is easy to know that the restriction of the 1-form $\omega_1$ to the set $\Sigma_{c_1,c_2}=\{(x,y,w,z)\in\mathbb R^4\;|\;z=c_1,w=c_2\}$ still has the form $\omega_1|_{\Sigma_{c_1,c_2}}=e^{-x}dy+2ye^{-x}dx$. Let $I_3=y^{\frac{1}{2}}e^x,$ we can get that 
$$dI_3\wedge\omega_1|_{\Sigma_{c_1, c_2}}=0, $$
which means that $I_3$ is also a first integral of $X_H.$ So, $X_H$ is completely integrable in $\mathbb R^4.$
\begin{remark}
 Function $f_1=1$ is non‑trivial because the Reeb field $E$ does not vanish; indeed $\{f_0,f_1\}=-1\neq0$ shows that the constant function $1$ is not a Casimir, in contrast to the Poisson/symplectic case.
\end{remark}

\end{example}

\subsection{Contact Hamiltonian systems}

Let $(M,\eta)$ be a contact manifold of dimension $2n+1$, with Reeb vector field
$R$. 
One can define an isomorphism of \(C^\infty(M)\)-modules given by
\[b : \mathfrak{X}(M) \ni X \mapsto X\contract d\eta + (X\contract\eta)\eta \in \Omega^1(M).\]
 Reeb vector field \(R\) can be expressed as \(R = b^{-1}(\eta)\). Moreover, to each function \(f \in C^\infty(M)\) one can associate a (contact) Hamiltonian vector field \(X_f\) given by 
\[b(X_f) = df - (Rf + f)\eta.\]
The contact manifold \((M, \eta)\) has a Jacobi structure \((\Lambda, E)\), where \(E = -R\) and the bivector \(\Lambda\) is given by 
\[\Lambda(\alpha, \beta) = -\mathrm{d}\eta(b^{-1}(\alpha), b^{-1}(\beta)).\]
The Jacobi bracket \(\{ \cdot , \cdot \} : C^{\infty}(M) \times C^{\infty}(M) \to C^{\infty}(M)\) is
\[ \{f, g\} = -\mathrm{d}\eta\left(b^{-1}df, b^{-1}dg\right) - fR(g) + gR(f),\]
and Hamiltonian vector
fields  are given by
\[
X_f=\Lambda^\sharp(df)-fR.
\]

According to Theorem~\ref{thm:jacobi-main}, we can state the following result for Hamiltonian systems on contact manifolds:
\begin{theorem}\label{ConC}
A Hamiltonian system on a 2$n+1$-dimensional  contact manifold that admits a Jacobi $C^\infty$-structure is exactly integrable by solving a
sequence of $2n$ completely integrable Pfaffian equations, requiring only $2n-1$ additional functions beyond the Hamiltonian itself.
\end{theorem}

    \begin{example}\label{ex:contact-Jacobi-Cinf}
Consider a three-dimensional contact manifold $(\mathbb R^3,\eta)$ with contact form $\eta = dz - y\,dx$ in coordinates $(x,y,z)$. The associated Jacobi structure is given by the Reeb vector field $R = \frac{\partial}{\partial z}$ and the bivector field
\[
\Lambda = \frac{\partial}{\partial x} \wedge \frac{\partial}{\partial y} + y \frac{\partial}{\partial z} \wedge \frac{\partial}{\partial y}.
\]
The Jacobi bracket for functions $f,g \in C^\infty(M)$ is:
\[
\{f,g\} = \Lambda(df,dg) + f R(g) - g R(f).
\]
Explicitly in coordinates:
\[
\{f,g\} = \left(\frac{\partial f}{\partial x}\frac{\partial g}{\partial y} - \frac{\partial f}{\partial y}\frac{\partial g}{\partial x}\right) 
+ y\left(\frac{\partial f}{\partial z}\frac{\partial g}{\partial y} - \frac{\partial f}{\partial y}\frac{\partial g}{\partial z}\right)
+ f\frac{\partial g}{\partial z} - g\frac{\partial f}{\partial z}.
\]

Consider a Hamiltonian function $f_2:=H=z,$ and take the auxiliary function $f_1=1$.

Compute their Jacobi bracket:
\begin{align*}
\{f_1, f_2\} &= \{1, z\} \\
&= \left(\frac{\partial 1}{\partial x}\frac{\partial z}{\partial y} - \frac{\partial 1}{\partial y}\frac{\partial z}{\partial x}\right) 
+ y\left(\frac{\partial 1}{\partial z}\frac{\partial z}{\partial y} - \frac{\partial 1}{\partial y}\frac{\partial z}{\partial z}\right)
+ 1\frac{\partial z}{\partial z} - z\frac{\partial 1}{\partial z} \\
&= (0\cdot 0 - 0\cdot 0) + y(0\cdot 0 - 0\cdot 1) + 1\cdot 1 - z\cdot 0 \\
&= 1.
\end{align*}
Thus we have:
\[
\{f_1, f_2\} = f_1.
\]

It is clear that $R\in <X_H,X_1>=<X_H,R>$. Hence we can see that $(f_1)$ form a Jacobi $C^\infty$-structure, from Theorem \ref{ConC}, we know that the Hamiltonian vector field $X_{H}$ can be integrated explicitly.
In fact, we can calculate that $X_H=y\frac{\partial}{\partial y}-z\frac{\partial}{\partial z}$. 
Taking a vector field $Y=\frac{\partial}{\partial x}$ and a volume form $\Omega=dx\wedge dy\wedge dz$,  we can get that
\begin{align*}
\omega_2:=X_{f_1}\contract X_H\contract\Omega=ydx,\quad \omega_1=Y\contract X_H\contract\Omega=-ydz-zdy.
\end{align*}
It is clear that 
$$dx\wedge\omega_2=0,\quad X_H\contract\omega_2=0$$
thus that $I_2=x$ is a first integral of $X_H.$ Moreover, the restriction of the 1-form $\omega_1$ to the set $\Sigma_c=\{(x,y,z)\in\mathbb R^3\;|\;x=c\}$ still has the form $\omega_1|_{\Sigma_c}=-ydz-zdy$. Let $I_1=yz,$ we know that 
$$dI_1\wedge\omega_1|_{\Sigma_c}=0, $$
which means that $I_1$ is also a first integral of $X_H.$ 
\end{example}

\section{Conclusions}

Poisson $\mathcal C^\infty$-structures provide a constructive framework for the exact integration of Hamiltonian systems on symplectic manifolds: whenever $\dim M - 2$ smooth functions close under the Poisson bracket with the Hamiltonian, the equations of motion can be integrated by a sequence of completely integrable Pfaffian equations. This differs from the Liouville paradigm (which demands $\tfrac{1}{2}\dim M$ independent first integrals in involution and yields a global torus foliation) and is closer in spirit to Mishchenko--Fomenko noncommutative integrability, though the emphasis here is on explicit solvability via Pfaffian equations rather than on global topology. The relationship with the solvable structures of \cite{kresicjuric2025} is complementary: both lead to integration via Pfaffian forms in the Liouville setting, but \cite{kresicjuric2025} harnesses conserved quantities to recover action--angle variables, while the present approach accommodates dynamically evolving generators at the cost of replacing quadratures by completely integrable Pfaffian equations.

The Jacobi $C^\infty$-structure (Section~\ref{sec:jacobicinf}) extends this framework to Jacobi manifolds, covering Poisson, locally conformally symplectic, and contact geometries in a unified way. The key new ingredient is a Reeb compatibility condition, which reduces to the Poisson case when $E=0$. Theorem~\ref{thm:jacobi-main} gives explicit integration via $m-1$ Pfaffian equations; Theorems~\ref{LCSC} and~\ref{ConC} specialize to LCS and contact manifolds, requiring $2n-2$ and $2n-1$ auxiliary functions, respectively. Contact geometry is especially notable as it enables exact integration on odd-dimensional manifolds, entirely outside the scope of classical Liouville theory.

Finally, among finite-dimensional reductions of the Vlasov equation, waterbag distributions are singled out as the natural class whose induced dynamics closes on a finite Poisson algebra, thereby admitting a Poisson $\mathcal C^\infty$-structure and exact Pfaffian integration. More general smooth distribution families fail to produce a compatible finite-dimensional Poisson algebra. Waterbag distributions thus emerge not merely as convenient approximations, but as canonical geometric objects characterized by invariance, symplectic consistency, and integrability.

\bibliographystyle{ieeetr}
\bibliography{references} 

@book{Arnold2,
  author    = {V. I. Arnold},
  title     = {Mathematical Methods of Classical Mechanics},
  edition   = {2nd},
  publisher = {Springer-Verlag},
  address   = {New York},
  year      = {1988}
}

@incollection{Arnold3,
  author    = {V. I. Arnold and A. B. Givental},
  title     = {Symplectic Geometry},
  booktitle = {Dynamical Systems IV},
  series    = {Encyclopaedia of Mathematical Sciences},
  volume    = {4},
  publisher = {Springer-Verlag},
  address   = {Berlin},
  year      = {2001}
}

@book{Bourbaki,
  author    = {N. Bourbaki},
  title     = {Elements of Mathematics: Linear and Multilinear Algebra},
  volume    = {2},
  chapter   = {2},
  publisher = {Addison-Wesley},
  year      = {1973}
}

@article{Cayley,
  author    = {A. Cayley},
  title     = {Sur les d{\'e}terminants gauches},
  journal   = {J. Reine Angew. Math.},
  volume    = {38},
  pages     = {93--96},
  year      = {1849}
}

@book{Vein,
  author    = {R. Vein and P. Dale},
  title     = {Determinants and Their Applications in Mathematical Physics},
  publisher = {Springer-Verlag},
  address   = {New York},
  year      = {1999}
}

@book{Gantmacher,
  author    = {F. R. Gantmacher},
  title     = {The Theory of Matrices},
  volume    = {1},
  publisher = {Chelsea Publishing Company},
  address   = {New York},
  year      = {1959}
}

@article{Zhao1,
  author    = {X. F. Zhao and Y. Li},
  title     = {Conserved Quantities, Symmetries of (Almost-){H}amiltonian and {L}agrangian Systems},
  journal   = {J. Geom. Anal.},
  volume    = {35},
  number    = {9},
  pages     = {268},
  year      = {2025}
}

@article{sardon1,
  author  = {M. de Le{\'o}n and C. Sard{\'o}n},
  title   = {Cosymplectic and contact structures for time-dependent and dissipative {H}amiltonian systems},
  journal = {J. Phys. A: Math. Theor.},
  volume  = {50},
  number  = {25},
  pages   = {255205},
  year    = {2017},
  doi     = {10.1088/1751-8121/aa6f2a}
}

@article{sardon2,
  author  = {O. Esen and M. de Le{\'o}n and C. Sard{\'o}n and M. Zajac},
  title   = {{H}amilton--{J}acobi formalism on locally conformally symplectic manifolds},
  journal = {J. Math. Phys.},
  volume  = {62},
  number  = {3},
  pages   = {032902},
  year    = {2021},
  doi     = {10.1063/5.0037553}
}

@article{Burby2023VariableMoment,
  author       = {J. W. Burby},
  title        = {Variable‐moment fluid closures with Hamiltonian structure},
  journal      = {Sci. Rep.},
  year         = {2023},
  volume       = {13},
  pages        = {18286},
  doi          = {10.1038/s41598-023-45416-5},
  url          = {https://www.nature.com/articles/s41598-023-45416-5},
}

@article{ScovelWeinstein1994,
  author       = {C. Scovel and A. Weinstein},
  title        = {Finite‐dimensional Lie‐Poisson approximations to Vlasov–Poisson equations},
  journal      = {Commun. Pure Appl. Math.},
  year         = {1994},
  volume       = {47},
  number       = {5},
  pages        = {683--709},
  doi          = {10.1002/cpa.3160470505},
  url          = {https://doi.org/10.1002/cpa.3160470505},
}

@article{sardon3,
  author  = {O. Esen and M. de Le{\'o}n and M. Lainz and C. Sard{\'o}n and M. Zaj{\'a}c},
  title   = {Reviewing the geometric {H}amilton--{J}acobi theory concerning {J}acobi and {L}eibniz identities},
  journal = {J. Phys. A: Math. Theor.},
  volume  = {55},
  number  = {40},
  pages   = {403001},
  year    = {2022},
  doi     = {10.1088/1751-8121/ac8c2e}
}

@misc{kresicjuric2025,
      title={Solvable Structures for Hamiltonian Systems},
      author={S. Kresic-Juric and C. Muriel and A. Ruiz},
      year={2025},
      eprint={2504.02189},
      archivePrefix={arXiv},
      primaryClass={math-ph},
      url={https://arxiv.org/abs/2504.02189},
}

@Article{pancinf-sym,
  author   = {A. J. Pan-Collantes and A. Ruiz and C. Muriel and J. L. Romero},
  title    = {$\mathcal{C}^{\infty}$-symmetries of distributions and integrability},
  journal  = {J. Differ. Equ.},
  year     = {2023},
  volume   = {348},
  pages    = {126-153},
  issn     = {0022-0396},
  doi      = {https://doi.org/10.1016/j.jde.2022.11.051},
  url      = {https://www.sciencedirect.com/science/article/pii/S0022039622006982},
}

@Article{pancinf-struct,
  author    = {A. J. Pan-Collantes and A. Ruiz and C. Muriel and J. L. Romero},
  title     = {$\mathcal{C}^{\infty}$-structures in the integration of involutive distributions},
  journal   = {Phys. Scr.},
  year      = {2023},
  volume    = {98},
  number    = {8},
  pages     = {085222},
  month     = {jul},
  doi       = {10.1088/1402-4896/ace403},
  fjournal  = {Phys. Scr.},
  publisher = {IOP Publishing},
  url       = {https://dx.doi.org/10.1088/1402-4896/ace403},
}

@Article{pan23integration,
  author    = {A. J. Pan-Collantes and C. Muriel and A. Ruiz},
  title     = {Integration of differential equations by $\mathcal{C}^{\infty}$-structures},
  journal   = {Mathematics},
  year      = {2023},
  volume    = {11},
  number    = {18},
  pages     = {3897},
  publisher = {MDPI},
}

@article{Mishchenko1978Euler,
  author  = {A. S. Mishchenko and A. T. Fomenko},
  title   = {{E}uler Equations on Finite-Dimensional {L}ie Groups},
  journal = {Math. USSR-Izv.},
  year    = {1978},
  volume  = {12},
  number  = {2},
  pages   = {371--389},
  doi     = {10.1070/IM1978v012n02ABEH001859}
}

@article{Mishchenko1978Generalized,
  author  = {A. S. Mishchenko and A. T. Fomenko},
  title   = {Generalized {L}iouville Method of Integration of {Hamiltonian} Systems},
  journal = {Funct. Anal. Appl.},
  year    = {1978},
  volume  = {12},
  number  = {2},
  pages   = {113--121},
  doi     = {10.1007/BF01076254}
}

@article{Liouville,
    author  = {J. Liouville},
    title   = {Note sur l'int{\'e}gration des {\'e}quations diff{\'e}rentielles de la dynamique},
    journal = {J. Math. Pures Appl.},
    volume  = {20},
    pages   = {137--138},
    year    = {1855}
}

@article{Fuchssteiner,
    author = {B. Fuchssteiner},
    title = {Application of hereditary symmetries to nonlinear evolution equations},
    journal = {Nonlinear Anal.},
    volume={3},
    pages={849–862},
    year = {1979}
}

@article{Magri,
  title     = {A simple model of the integrable Hamiltonian equation},
  author    = {F. Magri},
  journal   = {J. Math. Phys.},
  volume    = {19},
  number    = {5},
  pages     = {1156--1162},
  year      = {1978},
  publisher = {American Institute of Physics}
}

@article{Olver,
  title     = {Evolution equations possessing infinitely many symmetries},
  author    = {P. J. Olver},
  journal   = {J. Math. Phys.},
  volume    = {18},
  number    = {6},
  pages     = {1212--1215},
  year      = {1977},
  publisher = {American Institute of Physics}
}

@article{Azuaje2,
  title     = {{L}ie integrability by quadratures for symplectic, cosymplectic, contact and cocontact {H}amiltonian systems},
  author    = {R. Azuaje},
  journal   = {Rep. Math. Phys.},
  volume    = {93},
  number    = {1},
  pages     = {37--56},
  year      = {2024},
  publisher = {Elsevier}
}

@article{Carinena4,
  title     = {Geometry of {L}ie integrability by quadratures},
  author    = {J. F. Cariñena and F. Falceto and J. Grabowski and M. F. Rañada},
  journal   = {J. Phys. A: Math. Theor.},
  volume    = {48},
  number    = {21},
  pages     = {215206},
  year      = {2015},
  publisher = {IOP Publishing}
}

@article{Grabowska,
  title     = {Solvable {L}ie algebras of vector fields and a {L}ie's conjecture},
  author    = {K. Grabowska and J. Grabowski},
  journal   = {SIGMA},
  volume    = {16},
  pages     = {065},
  year      = {2020},
  publisher = {SIGMA}
}

@article{Kozlov,
  title     = {The {E}uler--{J}acobi--{L}ie integrability theorem},
  author    = {V. V. Kozlov},
  journal   = {Regul. Chaotic Dyn.},
  volume    = {18},
  number    = {4},
  pages     = {329--343},
  year      = {2013},
  publisher = {Pleiades Publishing}
}

@article{Carlet,
  title     = {Bihamiltonian cohomology of {KdV} brackets},
  author    = {G. Carlet and H. Posthuma and S. Shadrin},
  journal   = {Commun. Math. Phys.},
  volume    = {341},
  number    = {3},
  pages     = {805--819},
  year      = {2016},
  publisher = {Springer}
}

@article{Carlet2,
  title     = {The bi-{H}amiltonian cohomology of a scalar {P}oisson pencil},
  author    = {G. Carlet and H. Posthuma and S. Shadrin},
  journal   = {Bull. Lond. Math. Soc.},
  volume    = {48},
  number    = {4},
  pages     = {617--627},
  year      = {2016},
  publisher = {Wiley Online Library}
}

@article{Dubrovin,
  title     = {Bihamiltonian cohomologies and integrable hierarchies {II}: The tau structures},
  author    = {B. Dubrovin and S.-Q. Liu and Y. Zhang},
  journal   = {Commun. Math. Phys.},
  volume    = {361},
  number    = {2},
  pages     = {467--524},
  year      = {2018},
  publisher = {Springer}
}

@article{Lorenzoni2,
  title     = {A cohomological construction of integrable hierarchies of hydrodynamic type},
  author    = {P. Lorenzoni and F. Magri},
  journal   = {Int. Math. Res. Not.},
  volume    = {2005},
  number    = {34},
  pages     = {2087--2100},
  year      = {2005},
  publisher = {Oxford University Press}
}

@article{Zhao2,
  title   = {The reduction, first integral and {KAM} tori for $n$-dimensional volume-preserving systems},
  author  = {X. F. Zhao and Y. Li},
  journal = {J. Dyn. Differ. Equ.},
  volume  = {37},
  pages   = {539--558},
  year    = {2025},
  doi     = {10.1007/s10884-023-10293-z}
}

@article{Zhao3,
  title     = {Completely integrable system with {J}acobi multipliers and its {KAM} stability},
  author    = {X. F. Zhao},
  journal   = {Discrete Contin. Dyn. Syst.},
  volume    = {45},
  number    = {4},
  pages     = {1870--1890},
  year      = {2025}
}

@article{Duistermaat1980,
  title     = {On global action-angle coordinates},
  author    = {J. J. Duistermaat},
  journal   = {Commun. Pure Appl. Math.},
  volume    = {33},
  number    = {6},
  pages     = {687--706},
  year      = {1980},
  publisher = {Wiley Online Library}
}

@article{Nehorosev1972,
  title     = {Action-angle variables, and their generalizations},
  author    = {N. N. Nehorošev},
  journal   = {Trans. Moscow Math. Soc.},
  volume    = {26},
  pages     = {181--198},
  year      = {1972},
  note      = {in Russian}
}

@article{Zhao5,
    author  = {X. Zhao},
    title   = {Canonoid transformation and master symmetries of {H}amiltonian systems on locally conformal symplectic manifold},
    journal = {J. Math. Phys.},
    volume  = {66},
    number  = {8},
    pages   = {082704},
    year    = {2025},
    doi     = {10.1063/5.0255345},
    url     = {https://doi.org/10.1063/5.0255345},
}

@article{Azuaje,
    author  = {R. Azuaje and X. Zhao},
    title   = {Canonical and canonoid transformations for {H}amiltonian systems on locally conformal symplectic manifolds},
    journal = {J. Geom. Phys.},
    volume  = {222},
    pages   = {105761},
    year    = {2026},
    doi     = {10.1016/j.geomphys.2026.105761},
    url     = {https://www.sciencedirect.com/science/article/pii/S0393044026000112},
}
\end{document}